\documentclass[10pt,a4paper,english]{amsart}
\usepackage{hyperref}
\usepackage{babel}
\newfont{\ffont}{cmr10}
\allowdisplaybreaks
\usepackage[numbers]{natbib}
\usepackage{amsfonts, amsmath, wasysym}
\usepackage{amssymb, amsthm}
\usepackage{mathrsfs}
\usepackage{verbatim}
\usepackage[usenames,dvipsnames]{color}
\usepackage{amsmath, amssymb, amsfonts, amsbsy, amsthm,
latexsym,graphicx}
\usepackage{cancel}
\usepackage{pifont}
\usepackage{enumerate}
\usepackage{enumitem}
\topmargin by -\headsep \textheight 8.9in \oddsidemargin 0pt
\evensidemargin \oddsidemargin \marginparwidth 0.5in

\allowdisplaybreaks

\newcommand{\cP}{\mathcal{P}}
\newcommand{\cF}{\mathcal{F}}
\newcommand{\cS}{\mathcal{S}}

\newcommand{\id}{\mathrm{id}}

\newcommand{\al}{\alpha}

\newcommand{\ga}{\gamma}

\newcommand{\ka}{\kappa}

\newcommand{\si}{\sigma}

\newcommand{\Del}{\Delta}

\newcommand{\R}{\mathbb{R}}
\newcommand{\C}{\mathbb{C}}
\newcommand{\N}{\mathbb{N}}

\newcommand{\Z}{\mathbb{Z}}
\newcommand{\E}{\mathbb{E}}

\newcommand{\bP}{\mathbb{P}}

\newcommand{\eps}{\varepsilon}
\newcommand{\ti}{\times}

\newcommand{\argmin}{\mathrm{argmin}}

\newcommand{\Id}{\mathrm{Id}}
\newcommand{\mixedell}[2]{\ell_#1^b(\ell_#2^d)}

\newcommand{\T}{\ensuremath{\mathbb{T}^d}}

\newcommand{\card}{\mathrm{card}}

\makeatletter

\newcommand{\Rmnum}[1]{\expandafter\@slowromancap\romannumeral #1@}
\makeatother

\makeatletter
\newcommand{\specificthanks}[1]{\@fnsymbol{#1}}
\makeatother

\newcommand{\xspace}{\hbox{\kern-2.5pt}}

\begin{document}

\newtheorem{theorem}{Theorem}[section]
\newtheorem{lemma}[theorem]{Lemma}
\newtheorem{corollary}[theorem]{Corollary}
\newtheorem{proposition}[theorem]{Proposition}

\newcommand{\HR}[1]{{\color{Magenta}{#1}}}
\newcommand{\todo}[1]{\marginpar{\textcolor{red}{TODO}}\textcolor{red}{#1}}

\theoremstyle{definition}
\newtheorem{definition}[theorem]{Definition}
\newtheorem{xca}[theorem]{Exercise}
\newtheorem{design-crit}[theorem]{Design Criterion}
\newtheorem{remark}[theorem]{Remark}

\theoremstyle{remark}
\newtheorem{example}[theorem]{Example}

\title[Gelfand numbers, structured sparsity, Besov space embeddings]{Gelfand numbers related to structured sparsity and Besov space embeddings with small mixed smoothness}

\author{Sjoerd Dirksen}
\address{RWTH Aachen University, Lehrstuhl C f{\"u}r Mathematik (Analysis), Pontdriesch 10, 52062 Aachen, Germany}
\email{dirksen@mathc.rwth-aachen.de}

\author{Tino Ullrich}
\address{University of Bonn, Hausdorff Center for Mathematics, Endenicher Allee 62, 53115 Bonn, Germany} 
\email{tino.ullrich@hcm.uni-bonn.de}


\keywords{Gelfand numbers, $\ell_p(\ell_q)$-spaces, compressed sensing, block sparsity, sparsity-in-levels, Besov spaces with small mixed smoothness}

\begin{abstract}
We consider the problem of determining the asymptotic order of the Gelfand numbers of mixed-(quasi-)norm 
embeddings $\ell^b_p(\ell^d_q) \hookrightarrow \ell^b_r(\ell^d_u)$ given that $p \leq r$ and $q \leq u$, with emphasis on cases with $p\leq 1$ and/or $q\leq 1$. These cases turn out to be related to structured sparsity. We obtain sharp bounds in a number of interesting parameter constellations. 
Our new matching bounds for the Gelfand numbers of the embeddings of $\ell_1^b(\ell_2^d)$ and $\ell_2^b(\ell_1^d)$ into $\ell_2^b(\ell_2^d)$ imply optimality assertions for the recovery of 
block-sparse and sparse-in-levels vectors, respectively. In addition, we apply our sharp estimates 
for $\ell^b_p(\ell^d_q)$-spaces to obtain new two-sided estimates for the Gelfand numbers of
multivariate Besov space embeddings in regimes of small mixed smoothness. It turns out that in some particular cases these estimates 
show the same asymptotic behavior as in the univariate situation. In the remaining cases they differ at most by a $\log\log$ factor
from the univariate bound.
\end{abstract}

\maketitle

\section{Introduction}

\noindent Gelfand numbers are fundamental concepts in approximation theory and Banach space geometry. 
From a geometric point of view, the Gelfand numbers of an embedding between two (quasi-)normed spaces $X$ and $Y$, 
given by
$$c_m(\Id:X\to Y) = \inf\Big\{\sup_{x\in B_X\cap M} \|x\|_Y~:~M \hookrightarrow X \text{ closed subspace with } \text{codim}(M)<m\Big\},$$
quantify the compactness of the embedding $X\hookrightarrow Y$. From the perspective of approximation theory and signal 
processing, Gelfand numbers of embeddings naturally appear when considering the problem of the optimal recovery of an element $f \in X$ from few arbitrary \emph{linear} samples, where the recovery error is measured in $Y$. For instance, $f$ may be a function in a certain function space and the linear samples may consist of Fourier or wavelet coefficients of $f$ or simple function values. 
The $m+1$-th Gelfand number 
$c_{m+1}(\Id:X\to Y)$ is under mild assumptions equivalent to the \emph{$m$-th minimal worst-case recovery error} 
\begin{equation} 
\label{minwce}
 E_m(X,Y) := \inf_{S_m} \sup_{\|f\|_X \leq 1} \|f - S_m(f)\|_Y,
\end{equation}
where the infimum is taken over all maps of the form $S_m = \varphi_m \circ N_m$. Here 
$N_m: X \to \R^m$ is a linear and continuous measurement map and $\varphi_m: \R^m \to Y$ is an arbitrary (not necessarily linear) recovery map. 
More details on this connection can be found in Proposition \ref{gelfand_wce} and the discussion in Remark 
\ref{rem:Gelfand_width} below.\par
The aim of this paper is to understand the behavior of the Gelfand numbers of embeddings between $\ell_p^b(\ell_q^d)$-spaces, which are equipped with the mixed (quasi-)norm
\begin{equation}\label{mixednorm}
\|x\|_{\ell_p^b(\ell_q^d)} := \Big(\sum_{i=1}^b\Big(\sum_{j=1}^d |x_{ij}|^q\Big)^{p/q}\Big)^{1/p}.
\end{equation}
We put emphasis on cases where either $0<p\leq 1$ or $0<q\leq 1$ or both $0<p,q\leq 1$, which in particular 
means that we are confronted with parameter constellations where $\|\cdot\|_{\ell_p^b(\ell_q^d)}$ is only a quasi-norm 
and the corresponding unit ball is \emph{non-convex}. The $\ell_p^b(\ell_2^d)$- and $\ell_2^b(\ell_p^d)$-balls with $0<p \leq 1$ appear naturally in signal 
processing as models for (approximately) \emph{block-sparse} and \emph{sparse-in-levels} signals, respectively, see Subsection \ref{sec:sparsity} 
for further details. The study of Gelfand numbers of the embeddings of $\ell_p^b(\ell_2^d)$ and $\ell_2^b(\ell_p^d)$ into $\ell_2^{bd}$, which we conduct in this paper, yields general perfomance bounds for the optimal recovery of signals with structured sparsity, 
see Subsection \ref{SigProc}. On the other hand, $\ell_p^b(\ell_q^d)$-balls appear naturally in wavelet characterizations of multivariate Besov spaces with \emph{dominating mixed smoothness}. The study of embeddings between quasi-normed $\ell_p^b(\ell_q^b)$-spaces in this paper leads to new asymptotic bounds for Gelfand numbers of embeddings between the aforementioned Besov spaces in regimes of \emph{small smoothness}, see Subsection \ref{subsec:intro_Besov}. Interestingly, in some particular cases these estimates show the same asymptotic behavior as in the univariate situation. In the remaining cases they differ at most by a $\log\log$ factor from the univariate bound.

\subsection{Gelfand numbers and sparse recovery}
\label{SigProc}

We obtain matching bounds for the Gelfand numbers
$$c_m(\id:\ell_p^b(\ell_q^d)\to \ell_r^b(\ell_u^d))$$
for a number of interesting parameter ranges in this paper. The results in full detail can be found in Section \ref{sec:matchGel}. At this point, we wish to discuss only the most illustrative cases and highlight some interesting effects. In particular, we prove for $1 \leq m \leq bd$ and $0<p\leq 1$ 
\begin{equation}\label{f0}
 c_m(\id:\mixedell{p}{2} \to \ell_2^b(\ell_2^d)) \simeq_p \min\left\lbrace 1, \frac{\log(eb/m)+d}{m}\right\rbrace^{1/p-1/2}.
\end{equation}
For $1\leq m\leq bd$ and $0<q\leq 1$ we show that
\[ 
 c_m(\id:\mixedell{2}{q} \to \ell_2^b(\ell_2^d)) \simeq_q c_{\lfloor m/b \rfloor}(\id:\ell_q^d\to \ell_2^d) \simeq_q \min\left\lbrace1,\frac{b \log(ebd/m)}{m}\right\rbrace^{1/q-1/2}.
\] 
Comparing these two bounds, we clearly see that the roles of the outer and inner $\ell_p$-space are asymmetric. Taking $p=1$ and $q=1$, we find the following result as a special case.
\begin{theorem}
\label{thm:GelSimple}
Let $b,d\in \N$ and $m\leq bd$. Then,
\begin{align}
c_m(\id: \ell_1^b(\ell_2^d) \to \ell_2^b(\ell_2^d)) & \simeq \min\Big\{1,\frac{\log(eb/m)+d}{m}\Big\}^{1/2} \label{eqn:l1l2}\,, \\
c_m(\id:\ell_2^b(\ell_1^d) \to \ell_2^b(\ell_2^d)) & \simeq \min\Big\{1,\frac{b\log(ebd/m)}{m}\Big\}^{1/2}. \label{eqn:l2l1} 
\end{align}
\end{theorem}
Thanks to the equivalence between the Gelfand numbers in \eqref{eqn:l1l2} and \eqref{eqn:l2l1} and the minimal 
worst-case recovery errors $E_m(\ell_1^b(\ell_2^d), \ell_2^b(\ell_2^d))$ and $E_m(\ell_2^b(\ell_1^d), \ell_2^b(\ell_2^d))$, 
respectively, these bounds have interesting implications for the recovery of structured sparse signals. 
To discuss this, we fix some terminology. Consider $x\in \R^{bd}$ and view it as a $bd$-dimensional block vector $(x_i)_{i=1}^b$, 
with $x_i \in \R^d$ for all $1\leq i\leq b$. As usual, $x$ is called $k$-sparse if it has at most $k$ nonzero entries. 
We say that $x$ is \emph{$s$-outer sparse} if at most $s$ blocks of $x$ are non-zero. We call $x$ \emph{$t$-inner sparse} if each block of $x$ has at most $t$ non-zero entries. These two forms of structured sparsity exactly correspond to the notion of block sparsity and a special case of sparsity-in-levels, respectively, in the literature. We refer to Section~\ref{sec:sparsity} for a detailed discussion. Corresponding to outer and inner sparsity, we define two different best $s$-term approximation errors:
\begin{equation}
\label{eqn:beststermOI}
\begin{split}
\si_s^{\operatorname{outer}}(x)_{\ell_p^b(\ell_q^d)} & = \inf\{\|x-z\|_{\ell_p^b(\ell_q^d)} \ : \ z \ \operatorname{is} \ s\operatorname{-outer \ sparse}\}, \\
\si_t^{\operatorname{inner}}(x)_{\ell_p^b(\ell_q^d)} & = \inf\{\|x-z\|_{\ell_p^b(\ell_q^d)} \ : \ z \ \operatorname{is} \ t\operatorname{-inner \ sparse}\}.
\end{split}
\end{equation} 
\begin{corollary}
\label{cor:stabRecIntro}
Let $A\in\R^{m\ti bd}$ be a measurement matrix and $\Del:\R^m \rightarrow \R^{bd}$ be any reconstruction map. Suppose that for all $x\in \R^{bd}$,
\begin{equation}
\label{eqn:stabBlock}
\|x-\Del(Ax)\|_{\ell_2^b(\ell_2^d)} \leq C\frac{\si_s^{\operatorname{outer}}(x)_{\ell_1^b(\ell_2^d)}}{\sqrt{s}}.
\end{equation}
Then, there exist constants $c_1,c_2>0$ depending only on $C$ so that if $s>c_2$, then
$$m\geq c_1(s\log(eb/s)+sd).$$
On the other hand, if for all $x\in \R^{bd}$,
\begin{equation}
\label{eqn:stabInner}
\|x-\Del(Ax)\|_{\ell_2^b(\ell_2^d)} \leq D\frac{\si_t^{\operatorname{inner}}(x)_{\ell_2^b(\ell_1^d)}}{\sqrt{t}},
\end{equation}
then there exist constants $c_1,c_2>0$ depending only on $D$ such that
$$m\geq c_1 bt\log(ed/t),$$
provided that $t>c_2$.
\end{corollary}
To discuss the consequences of this result for structured sparse recovery, let us recall the following known results. We refer to Section~\ref{sec:sparsity} for a more extensive discussion of the literature. Let $B$ be an $m\ti bd$ Gaussian matrix, i.e., a matrix with i.i.d.\ standard Gaussian entries and set $A=(1/\sqrt{m})B$. For any $0<p,q\leq \infty$, let 
$$\Del_{\ell_p(\ell_q)}(Ax) = \argmin_{z\in \R^{bd} \ : \ Az=Ax}\|z\|_{\ell_p^b(\ell_q^d)}$$ 
be the reconstruction via $\ell_p(\ell_q)$-minimization. It is known that if the number of measurements $m$ satisfies $m\gtrsim s\log(eb/s)+sd$, 
then with high probability the pair $(A,\Del_{\ell_1(\ell_2)})$ satisfies the stable reconstruction property in (\ref{eqn:stabBlock}) (see \cite{ElM09} and also the discussion in \cite[Section 2.1]{ADR14}). Moreover, this number of measurements cannot be reduced: \cite[Theorem 2.4]{ADR14} shows that if for a given $A\in \R^{m\ti bd}$ one can recover every 
$s$-outer sparse vector $x$ exactly from $y=Ax$ via $\Del_{\ell_1(\ell_2)}$ then necessarily $m\gtrsim s\log(eb/s)+sd$. However, one could still hope to further reduce the required number of measurements for outer sparse recovery by picking a better pair $(A,\Del)$. Corollary~\ref{cor:stabRecIntro} shows that this is not possible if we want the same stable reconstruction guarantee as in (\ref{eqn:stabBlock}). In this sense, the pair $(A,\Del_{\ell_1(\ell_2)})$ with $A=(1/\sqrt{m})B$ is \emph{optimal}.\par
It is worthwhile to note that the pair $(A,\Del_{\ell_1(\ell_1)})$, $A=(1/\sqrt{m})B$, with high probability satisfies the stable reconstruction guarantee (\cite[Theorem 1.2 and Section 1.3]{CRT06}, see also \cite{CaT06,do06_2} and \cite[Theorem 9.13]{FoR13})
\begin{equation}
\label{eqn:stab11}
\|x-\Del_{\ell_1(\ell_1)}(Ax)\|_{\ell_2^b(\ell_2^d)} \leq D\frac{\si_k(x)_{\ell_1^b(\ell_1^d)}}{\sqrt{k}},
\end{equation}
whenever $m\gtrsim k\log(ebd/k)$, where 
$$\si_k(x)_{\ell_1^b(\ell_1^d)} = \inf\{\|x-z\|_{\ell_1^b(\ell_1^d)} \ : \ z \ \operatorname{is} \ k\operatorname{-sparse}\}.$$
As a consequence (by taking $k=sd$), \eqref{eqn:stabBlock} is satisfied under the condition $m\gtrsim sd\log(eb/s)$. In conclusion, if we are only interested in stability of our reconstruction method with respect to outer sparsity (instead of sparsity) defects, then we can achieve this with fewer measurements (in a best worst-case scenario) by using $\Del_{\ell_1(\ell_2)}$ as a reconstruction map instead of $\Del_{\ell_1(\ell_1)}$. In the case of reconstruction of inner-sparse vectors, one might expect that a similar phenomenon occurs: if we are only interested in stability with respect to inner sparsity (instead of sparsity) defects, then we might be able to reconstruct from fewer measurements by using a reconstruction map that is especially suited for inner sparse recovery (a potential candidate would be $\Del_{\ell_2(\ell_1)}$, in analogy with the outer sparse case). Somewhat surprisingly, Corollary~\ref{cor:stabRecIntro} shows that this is not the case. By taking $k=bt$ in \eqref{eqn:stab11} and using that every $t$-inner sparse vector is $bt$-sparse (but not vice-versa), one deduces that the pair $(A,\Del_{\ell_1(\ell_1)})$, $A=(1/\sqrt{m})B$, with high probability satisfies \eqref{eqn:stabInner} if $m\gtrsim bt\log(ed/t)$. This is optimal by Corollary~\ref{cor:stabRecIntro}.    

\subsection{Besov space embeddings with small mixed smoothness}\label{subsec:intro_Besov}
As a second application of our work, we are interested in proving upper and lower bounds for the Gelfand numbers of compact embeddings between Besov spaces
\begin{equation}\label{emb2_0}
    \text{Id}:S^{r_0}_{p_0,q_0}B([0,1]^d) \to S^{r_1}_{p_1,q_1}B([0,1]^d)\,,
\end{equation}
where $0<p_0 \leq p_1 \leq \infty$, $0<q_0, q_1 \leq \infty$ and $r_0-r_1>1/p_0-1/p_1$.
By using a well-known hyperbolic wavelet discretization machinery, see for instance \cite{Vyb06},
we can transfer the problem to bounding the Gelfand numbers of 
$$\text{id}:s^{r_0}_{p_0,q_0}b([0,1]^d) \to s^{r_1}_{p_1,q_1}b([0,1]^d),$$ 
where (see \eqref{sequBOmega} below for details)
\begin{equation*}
    s^r_{p,q}b := \Big\{\lambda=\{\lambda_{\bar{j},\bar{k}}\}_{\bar{j}, \bar{k}} \subset \C~:~
    \|\lambda\|_{s^r_{p,q}b}:=\Big[\sum\limits_{\bar{j}\in \N_0^d} 2^{(r-1/p)q\|\bar{j}\|_1}
    \Big(\sum\limits_{\bar{k}\in A_{\bar{j}}}|\lambda_{\bar{j},\bar{k}}|^p\Big)^{q/p}\Big]^{1/q}<\infty\Big\}
\end{equation*}
and $A_{\bar{j}} = ([0,2^{j_1}-1]\times \cdots \times [0,2^{j_d}-1]) \cap \Z^d$\,. We immediately observe that $s^r_{p,q}b$ has an $\ell_q(\ell_p)$-structure. After a suitable block decomposition we use the subadditivity of Gelfand numbers (see (S2) below) to reduce the problem to finite-dimensional blocks. This strategy has for instance been pursued in the recent paper \cite{Kien16}, where the case $p_0=q_0>1$ has been treated. In our work, we are interested in the case $\min\{p_0,q_0\} \leq 1$, $\max\{p_1,q_1\} \leq 2$, which is connected to Theorem~\ref{thm:GelMain} below. As it turns out, there is a strong qualitative difference in the behavior of the Gelfand numbers depending on whether one is in the ``large mixed smoothness'' ($r_0-r_1>1/q_0-1/q_1$) or ``small mixed smoothness'' ($r_0-r_1\leq 1/q_0-1/q_1$) regime. In the large smoothness regime, sharp bounds for the Gelfand numbers can be readily derived: using almost literally the arguments in \cite[Thms.\ 3.18, 3.19]{Vyb06} together with the well-known sharp estimates for the Gelfand numbers of the (non-mixed) embeddings $\id:\ell_u^n\to \ell^n_v$ where $0<u\leq 1$ and $0<u<v\leq 2$ (see \eqref{match_gelf} and also Theorem \ref{thm:GelMain}, \eqref{eqn:Foucartetal} below), one can show
\begin{equation}\label{gelflarge0}
    c_m(\id:S_{p_0,q_0}^{r_0}B([0,1]^d) \to S_{p_1,q_1}^{r_1}B([0,1]^d)) \simeq m^{-(r_0-r_1)}(\log m)^{(d-1)(r_0-r_1-1/q_0+1/q_1)}\quad,\quad m\geq 2\,.
\end{equation}
As we point out in Remark~\ref{large_s}, this argument breaks down in the regime of small smoothness. In this parameter regime we obtain sharp estimates in the following special case (see Theorem \ref{mainB} below) by applying our new sharp mixed norm result in Theorem \ref{thm:GelMain}, \eqref{eqn:GelMain1}. If $0<q_0 \leq 1$, $0< q_0<q_1\leq 2$, $r_1=0$ and $0<r_0=r<1/q_0-1/q_1$ then 
\begin{equation}\label{main1}
    c_m(\Id:S^r_{2,q_0}B([0,1]^d) \to S^0_{2,q_1}B([0,1]^d)) \simeq_{q_0,q_1,d,r} m^{-r}, \qquad m\in \N\,.
\end{equation}
In particular, choosing $q_1=2$ in \eqref{main1} we find as an important instance 
$$
c_m(\Id:S^r_{2,q_0}B([0,1]^d) \to L_2([0,1]^d)) \simeq_{q_0,d,r} m^{-r}\quad,\quad m\in \N,
$$
for $0<q_0 \leq 1$ and $0<r<1/q_0-1/2$. Let us emphasize that the smoothness parameter $r$ in \eqref{main1} is actually not 
restricted to small values as the term ``small smoothness'' suggests. Indeed, if $q_0$ is taken small, then $r$ is allowed to be large.\par 
Interestingly, the decay behavior in \eqref{main1} is typical for the univariate situation ($d=1$). In the $d$-variate setting we usually encounter 
asymptotic orders such as $m^{-r}(\log m)^{(d-1)\eta}$, where $\eta = \eta(p_0,q_0,p_1,q_1,r_0,r_1)$, see for 
instance \eqref{gelflarge0} above. That is, the dimension $d$ of the underlying Euclidean space 
enters the asymptotic bound exponentially as $m$ grows. Surprisingly,
this is no longer the case in the regime of small smoothness as \eqref{main1} indicates. Note, however, that the dimension $d$ is still hidden in 
the constants. This particular influence of the secondary parameters $q_0,q_1$ of the Besov spaces has been so far only observed 
for the best $m$-term widths 
\cite[Thm.\ 5.13, Cor.\ 6.1]{HaSi12} and recently also for the entropy numbers \citep{MaUl17}, 
\cite[Open Problem 6.4]{DTU16}. Interestingly, an asymptotic bound as given in \eqref{main1} can not hold for 
approximation numbers (linear widths) as the lower bounds in \cite[Theorem 4.47]{DTU16} indicate. This is of particular relevance from a complexity theoretic point of view, as it shows that allowing non-linear reconstruction maps $\varphi_m$ (rather than only linear ones) in \eqref{minwce} leads to a substantially faster decay of the minimal worst-case reconstruction error.\par
Note that there are several places in the literature where the term ``small smoothness'' occurs, beginning with \cite{Ka81}. Several authors 
discovered interesting phenomena of asymptotic characteristics in connection with small (mixed) Sobolev (or Besov) 
smoothness, see for instance the recent paper \cite{Te17} (and the discussion in Section 5 there) 
which deals with sparse (best $m$-term) approximation. However, the ``small smoothness effects'' discussed in 
\cite[Sect.\ 5]{Te17}, \cite[7.5, 8.4, 8.5]{DTU16}, or \cite[Thm.\ 1.2]{UlUl16} are of a different nature than 
the one discussed in the present paper.\par 
In the small smoothness regime we can only prove near-matching bounds if $p_0<p_1$ in \eqref{emb2_0}. For $0<p_0\leq 1$, $0<q_0\leq p_0\leq p_1 \leq q_1\leq 2$ and $1/p_0-1/p_1<r<1/q_0-1/q_1$ we prove that for constants $c,C>0$ depending only on $p_0,p_1,q_0,q_1,d$ and $r$,
\begin{equation}\label{main2}
 c m^{-r} \leq c_m(\Id:S^r_{p_0,q_0}B([0,1]^d) \to S^0_{p_1,q_1}B([0,1]^d)) 
 \leq C m^{-r}(\log\log_2 m)^{1/q_0-1/q_1}\,.
\end{equation}
By taking $p_1=q_1$ and $q_0<p_0$ we obtain \eqref{main2} also for target spaces $Y = L_{p_1}$ by embedding, see Corollaries \ref{L2}, \ref{endpoint}, \ref{nonsharp} 
and Remark \ref{SickelNguyen} below. We strongly conjecture that the $\log\log$-term in \eqref{main2} is not necessary. As we discuss in Remark~\ref{rem:diffInner}, it could be removed by proving sharp bounds in the spirit of Theorem \ref{thm:GelMain}, \eqref{eqn:GelMain1}, in the case of differing inner spaces. This seems to be rather involved for Gelfand numbers. However, it is at least possible to show that for the corresponding dyadic entropy numbers the lower bound is sharp \cite{MaUl17}. 

\subsection{Related work on Gelfand numbers}

As we have pointed out in Subsection~\ref{SigProc}, Gelfand numbers play an important role in \emph{compressive sensing} 
\cite{do06,FPR10,FoR13}, which deals with the recovery of high-dimensional, (approximately) sparse signals from 
linear measurements. In this theory, approximately sparse (also called \textit{compressible}) signals are modelled 
by vectors in non-convex $\ell_p^d$-balls with $0<p\leq 1$. Gelfand numbers of such non-convex balls measured in Euclidean space are of great importance 
since they give general performance bounds for sparse recovery methods. Beyond compressive sensing, compressibility assumptions also made their way into 
high-dimensional function recovery, concretely, the recovery of so-called ridge functions \cite{CoDaDeKePi2012,FoScVy2012,MaUlVy2015,KueMaUl2016}. 

Matching bounds for the Gelfand numbers for classical $\ell_p$-spaces with $p\geq 1$ (corresponding to $p=q$ and $r=u$ in our notation) are known for many parameter ranges, 
largely due to work in the late seventies and in the eighties. We refer to \cite[Section 4]{Vyb08} for a summary of these classical results and historical references. 
Let us only specifically mention the classical bound
\begin{equation}
c_m(\id:\ell_1^b(\ell_1^d) \to \ell_2^b(\ell_2^d)) \simeq \min\Big\{1,\frac{\log(ebd/m)}{m}\Big\}^{1/2}, \label{eqn:l1l1}
\end{equation}
which was obtained by Garnaev and Gluskin \cite{GaG84,Glu82} following important earlier advances of Kashin \cite{Kas77}. 
Note that one recovers this result from (\ref{eqn:l1l2}) by taking `$b=bd$', $d=1$. Some of the classical estimates were extended to the quasi-Banach range $p<1$, 
see \cite{FPR10,Vyb08}, where it was shown that for $0<p\leq 1$ and $p<q\leq 2$,
\begin{equation}\label{match_gelf}
c_m(\id:\ell_p^b(\ell_p^d) \to \ell_q^b(\ell_q^d)) \simeq_{p,q} \min\Big\{1,\frac{\log(ebd/m)}{m}\Big\}^{1/p-1/q}.
\end{equation}
The proof of this result exploits a connection between Gelfand widths for $\ell_p^b(\ell_p^d)$ and the recovery of sparse vectors. In a similar way, 
the proof of Theorem~\ref{thm:GelMain} exploits a connection with the recovery of vectors exhibiting structured sparsity.
The relation in \eqref{match_gelf} has been claimed earlier by Donoho in \cite{do06_2}. He used Carl's inequality for 
(quasi-)Banach spaces, which was not available at that time, to obtain the lower bound. This proof gap has been closed recently 
in \cite{HiKoVy16}\,.
\par
Let us now summarize some known results on Gelfand numbers for embeddings between $\ell_p^b(\ell^d_q)$-spaces. 
All works mentioned below (including Vasileva's recent paper \cite{Vas13}) concern parameter ranges that are complementary to our main result, Theorem~\ref{thm:GelMain}. In the nineties, Galeev \cite{Gal90,Gal96} and Izaak \cite{Iza94,Iza96} proved some (near-)matching bounds for Kolmogorov widths of mixed-norm $\ell_p(\ell_q)$-spaces. They applied their bounds to estimate the approximation numbers (or linear widths) of  embeddings of H\"{o}lder-Nikol'skii spaces of periodic functions with 
mixed smoothness $S^r_{p,\infty}B(\mathbb{T}^d)$ into $L_q$-spaces (see in particular the main theorem in \cite{Gal96} and \cite[Theorem 4]{Iza96}). 
By the well-known duality between Kolmogorov and Gelfand widths, i.e., $c_m(T:X \to Y)=d_m(T^\ast:Y^\ast \to X^\ast)$ 
for any Banach spaces $X,Y$, see \cite[11.7.7]{Pie80} and \cite[2.5.6]{Pie87}, their results can be phrased in terms of Gelfand 
widths as follows. If $1<q<\infty$ and $m<bd/2$, then by \cite[Thm.\ 2]{Gal90},
$$c_m(\id:\ell_q^b(\ell_2^d) \to \ell_1^b(\ell_{\infty}^d)) \simeq_q b^{1-\frac{1}{q}}.$$
Izaak \cite[Thm.\ 1]{Iza94} extended the result to $q=\infty$ and proved in case $m<bd/2$
$$\frac{b\sqrt{\log\log b}}{\log b} \lesssim c_m(\id:\ell_{\infty}^b(\ell_2^d) \to \ell_1^b(\ell_{\infty}^d)) \lesssim b.$$
The lower bound has been improved recently by Malykhin, Ryutin \cite{MaRu16} to the sharp bound
$$c_m(\id:\ell_{\infty}^b(\ell_2^d) \to \ell_1^b(\ell_{\infty}^d)) \simeq b.$$ 
This result can be used to prove the following sharp result
for Gelfand numbers of H\"older-Nikol'skij spaces $S^r_{p,\infty}B(\mathbb{T}^d)$, see \cite[pp. 136--140]{Gal96}, \cite{MaRu16}, \cite[Thm.\ 8.3]{ByUl16}. 
\begin{theorem}\label{satz:gelfandnikolskij} If $1<p<q<\infty$ and $r>1-1/q$, then
  	$$c_m(\Id:S^r_{p,\infty}B(\T) \to L_q(\T))\simeq_{p,q,r,d} \begin{cases}
  	\left(\frac{\log^{d-1} m}{m}\right)^{r-\frac{1}{2}+\frac{1}{q}}(\log m)^{\frac{d-1}{q}}&:\quad \frac{1}{p}+\frac{1}{q}<1,\;p\leq 2,\\
  	\left(\frac{\log^{d-1} m}{m}\right)^{r-\frac{1}{p}+\frac{1}{q}}(\log m)^{\frac{d-1}{q}}&:\quad 2\leq p<q.
  	\end{cases}$$
 \end{theorem}
\noindent The lower bound in Theorem \ref{satz:gelfandnikolskij} for the (larger) 
approximation numbers (linear widths) was proved recently by Malykhin and Ryutin. It closes a $\log\log$-gap in a result of 
Galeev \cite{Gal96} from 1996\,. 

Let us also mention the results from \cite[Thm.\ 3]{Iza96} and the remark following it. Let $m<bd/2$. First, 
$$c_m(\id:\ell_{\infty}^b(\ell_2^d) \to \ell_{\infty}^b(\ell_{\infty}^d)) \simeq 1;$$
Second, for $1\leq p\leq 2$,
$$c_m(\id:\ell_{\infty}^b(\ell_2^d) \to \ell_{p}^b(\ell_{p}^d))\simeq_p b^{\frac{1}{p}}d^{\frac{1}{p} - \frac{1}{2}}.$$
Third, if $2\leq p<\infty$ and in addition there exists a constant $\gamma>0$ such that $d\geq \gamma\log b$, then 
$$\frac{b^{\frac{1}{p}}}{(\log b)^{1-\frac{1}{p}}}\lesssim_{\gamma,p} c_m(\id:\ell_{\infty}^b(\ell_2^d) \to \ell_{p}^b(\ell_{p}^d)) \leq b^{\frac{1}{p}}.$$
Finally, we state a result of Vasileva from \cite{Vas13}. For $1\leq a\leq 2$ and $1\leq a\leq b\leq \infty$ we define
$$\lambda(a,b) = \min\Big\{\frac{\frac{1}{a}-\frac{1}{b}}{\frac{1}{a}-\frac{1}{2}},1\Big\}$$
if $a\neq 2$ and $\lambda(a,b)=1$ if $a=2$. Note that if $a<2$, then $\lambda(a,b)<1$ if and only if $b<2$.
\begin{theorem}[\cite{Vas13}]
Let $1<p\leq 2$,  $1<q\leq 2$, $p\leq r\leq \infty$, and $q\leq u\leq \infty$. Then there exists an $a=a(p,q)>0$ such that for all $m,b,d$ with $m\leq abd$ 
$$c_{m}(\id:\ell_p^b(\ell_q^d)\to\ell_r^b(\ell_u^d)) \simeq_{p,q,r,u} \Phi_0(m,b,d,p,q,r,u)$$
provided that one of the following conditions holds:
\begin{itemize}
\item[(i)] $u\geq 2$, $r\geq 2$, and
$$\Phi_0 = \min\{1,m^{-\frac{1}{2}} d^{1-\frac{1}{q}} b^{1-\frac{1}{p}}\}$$
\item[(ii)] $\lambda(q,u)\leq \lambda(p,r)$, $\lambda(q,u)<1$, and
\begin{align*}
\Phi_0 =\min\Big\{1, (m^{-\frac{1}{2}} d^{1-\frac{1}{q}} b^{1-\frac{1}{p}})^{\lambda(q,u)}, d^{\frac{1}{u}-\frac{1}{q}} (m^{-\frac{1}{2}} d^{\frac{1}{2}} b^{1-\frac{1}{p}})^{\lambda(p,r)}\Big\}
\end{align*}
\item[(iii)] $\lambda(q,u)\geq \lambda(p,r)$, $\lambda(p,r)<1$, and
\begin{align*}
\Phi_0 =\min\Big\{1, (m^{-\frac{1}{2}}d^{1-\frac{1}{q}} b^{1-\frac{1}{p}})^{\lambda(p,r)}, b^{\frac{1}{r}-\frac{1}{p}} (m^{-\frac{1}{2}} d^{1-\frac{1}{q}} b^{\frac{1}{2}})^{\lambda(q,u)}\Big\}.
\end{align*}
\end{itemize}
\end{theorem}

\subsection{Outline} After fixing notational conventions we introduce several basic concepts, such as Gelfand numbers and Gelfand widths and 
generalized notions of sparsity and compressibility in Section 2. In Section 3 we set up some machinery to prove upper 
bounds for Gelfand widths of bounded subsets in $(\R^n,\|\cdot\|_Y)$ via random subspaces based on Gordon's 
``escape through the mesh'' theorem. For the lower bounds we establish a connection to 
packing numbers of structured sparse vectors and give lower bounds for these packing numbers, see Sections~\ref{sec:lower_bound1} and \ref{sec:lower_bound2}. 
Section~\ref{sec:matchGel} establishes matching bounds for Gelfand numbers for embeddings between $\ell_p^b(\ell^d_q)$-spaces in several 
interesting parameter ranges. These results are applied in Section~\ref{sec:appl}, where we prove Corollary~\ref{cor:stabRecIntro} and establish new results for the Gelfand numbers of Besov space embeddings with small mixed smoothness.

\subsection{Notation} As usual $\N$ denotes the natural numbers, $\N_0:=\N\cup\{0\}$, $\Z$ denotes the integers, 
$\R$ the real numbers, and $\C$ the complex numbers. For
$a\in \R$ we denote $a_+ := \max\{a,0\}$. We write $\log$ for the natural logarithm. $\R^{m\times n}$ denotes the set of all $m\times n$-matrices with real entries and $\R^n$ denotes the Euclidean space. 
Vectors are usually denoted with $x,y\in \R^n$, sometimes we use $\bar{j},\bar{k} \in \N_0^d$ for multi-indices. For $0<p\leq \infty$ and $x\in \R^n$ we denote $\|x\|_p := (\sum_{i=1}^n
|x_i|^p)^{1/p}$ with the usual modification in the case $p=\infty$. 
If $X$ is a (quasi-)normed space, then $B_X$ denotes its unit ball and the (quasi-)norm
of an element $x$ in $X$ is denoted by $\|x\|_X$. For any $U\subset X$ we let $r_X(U)=\sup_{x\in U}\|x\|_X$ denote the radius of $U$. In case $X$ is a Banach space, $X^\ast$ denotes its dual. 
We will frequently use the quasi-norm constant, i.e., the smallest constant $\alpha_X$ satisfying
$$
    \|x+y\|_X \leq \alpha_X(\|x\|_X + \|y\|_X), \qquad \text{for all } x,y\in X.
$$
For a given $0<p\leq 1$ we say that $\|\cdot\|_X$ is a $p$-norm if 
$$
    \|x+y\|^p_X \leq \|x\|_X^p + \|y\|_X^p, \qquad \text{for all } x,y\in X.
$$
As is well known, any quasi-normed space can be equipped with an equivalent $p$-norm (for a certain $0<p\leq 1$), see 
\cite{Ao42,Ro57}. 
If $T:X\to Y$ is a continuous operator we write $T\in
\mathcal{L}(X,Y)$ and denote its operator (quasi-)norm by $\|T\|$. The symbol $X \hookrightarrow Y$ indicates that the
identity operator $\Id:X \to Y$ is continuous. For two sequences $(a_n)_{n=1}^{\infty},(b_n)_{n=1}^{\infty}\subset \R$  we
write $a_n \lesssim b_n$ if there exists an absolute constant $c>0$ such that $a_n \leq
c\,b_n$ for all $n$. We write $a_n \simeq b_n$ if both $a_n \lesssim b_n$ and
$b_n \lesssim a_n$. If $\alpha$ is a set of parameters, then we write $a_n \lesssim_{\alpha} b_n$ if there exists a constant $c_{\alpha}>0$ depending only on $\alpha$ such that $a_n \leq
c_{\alpha}\,b_n$ for all $n$.\par 
Let $b,d \in \N$. The $bd$-dimensional mixed space $\ell_p^b(\ell_q^d)$ is defined as the space of all $bd$-dimensional block vectors $(x_i)_{i=1}^b \in \R^{bd}$, where $x_i \in \R^d$ for all $1\leq i\leq b$, equipped with the mixed (quasi-)norm defined in \eqref{mixednorm}. Occasionally we will identify $x\in \R^{bd}$ with the $b\times d$ matrix whose rows are the vectors $x_i$. We always refer to the $\ell_p$-space supported on $[b]:=\{1,\ldots,b\}$ as the \emph{outer} space and to the $\ell_q$-space supported on $[d]$ as the \emph{inner} space. For any $S\subset[b]\times[d]$ and $x\in \R^{bd}$ we define $x$ to be the block vector with entries $(x_S)_{ij} = x_{ij}$ for $(i,j)\in S$, $(x_S)_{ij} = 0$ for $(i,j)\in S^c$. 

\section{Preliminaries}

\subsection{Gelfand numbers and worst-case errors}
\label{sec:GelEquiv}

We recall the definition of Gelfand numbers and the mild conditions under which Gelfand numbers can be considered as worst-case approximation/recovery errors. 
Let $X,Y$ be (quasi-)Banach spaces and $T:X\to Y$ be a continuous map. 

\begin{definition}\label{def:Gelf} Let $X,Y$ be two quasi-Banach spaces and $T \in \mathcal{L}(X,Y)$\,. For $m\in \N$ the $m$-th Gelfand number is defined as 
$$
    c_m(T) := \inf\{\|TJ^X_M\|~:~M \hookrightarrow X \text{ closed subspace with } \text{codim}(M)<m\}\,.
$$
Here $J_M^X$ is the natural injection of $M$ into $X$.
\end{definition}
We emphasize that the dimension of a subspace is a purely algebraic notion and makes sense also in the framework of 
quasi-Banach spaces. The codimension of a subspace can be defined as the dimension of the quotient space, see \cite[Sect.\ 1.40]{Ru91}. 

Similar to Kolmogorov and approximation numbers, Gelfand numbers are additive and multiplicative $s$-numbers, see \cite[Sect.\ 2.2]{Pie87}. 
The sequence $(c_m(T))_{m=1}^{\infty}$ satisfies the following properties. Let $X,Y,X_0,Y_0$ be quasi-Banach spaces and $Y$ be a $p$-Banach space. Let further 
$S,T\in\mathcal{L}(X,Y)$, $U \in \mathcal{L}(X_0,X)$, $R\in \mathcal{L}(Y,Y_0)$. Then the following properties hold:
\begin{description}
  \item[(S1)] $\|T\|_{\mathcal{L}(X,Y)}=c_1(T)\geq c_2(T)\geq\ldots\geq 0$\,,
  \item[(S2)] for all $m_1,m_2\in\N_0 $ 
  $$c_{m_1+m_2-1}^p(S+T)\leq c_{m_1}^p(S)+c_{m_2}^p(T)\,,$$
  \item[(S3)] for all $m\in \N$ 
  $$c_m(RTU) \leq \|R\|\cdot c_m(T)\cdot \|U\|\,,$$
  \item[(S4)] $c_m(T) = 0$ if $m>\mbox{rank}(T)\,,$
  \item[(S5)] $c_m(\mbox{id}:\ell_2^m\to \ell_2^m)=1\,.$
  \end{description}
If $X$ and $Y$ are Banach spaces then we may also define the $m$-the Gelfand number $c_m(T:X\to Y)$ by
\begin{equation} \label{def:Gelf_b}
 c_m(T) := \inf_{L_1,\dots,L_{m-1} \in X^*} \sup\{ \|T\,f\|_Y : f \in B_X, L_1(f)  = ... = L_{m-1}(f) = 0\}\,,
\end{equation}
which is the crucial property leading to the well-known relation between Gelfand numbers and minimal worst-case errors. 
In fact, if $X \hookrightarrow  Y$ and $T=\mbox{id}$ then we have the following well-known relation between Gelfand numbers
and the minimal worst-case errors defined in \eqref{minwce}. A first relation of this type has been given in 
\cite[Thm.\ 5.4.1]{TrWaWo88}, see also Remark \ref{rem:Gelfand_width} below.
The proof of the following statement is a straightforward modification of \cite[Thm.\ 10.4]{FoR13}\,. 
\begin{proposition}\label{gelfand_wce}Let $X$ and $Y$ be Banach spaces and $X\hookrightarrow Y$. Then
\[ 
 c_{m+1}(\Id:X \to Y) \leq E_m(X,Y) \leq 2c_{m+1}(\Id:X \to Y).
\]
\end{proposition}
Note, that \eqref{def:Gelf_b} and Proposition \ref{gelfand_wce} may fail in general for quasi-Banach spaces. 
\begin{remark}\label{rem:Gelfand_width} {\em (i)} There is also the related notion of \emph{Gelfand width}. 
The $m$-th Gelfand width of a bounded set $K$ in $Y$ is given by
\begin{equation}
\label{eqn:GwidthDef}
  c_m(K,Y) := \inf_{M \hookrightarrow Y, \text{codim}(M) \leq m} \sup\{ \|f\|_Y : f \in K\cap M\}.
\end{equation}
The terms Gelfand numbers and Gelfand widths are often used interchangeably in the literature although they are defined in a fundamentally different way. In particular, it is important to note that in the definition of the Gelfand numbers $c_m(\id:X\to Y)$ we take an infimum over subspaces $M$ of $X$, whereas in the definition of Gelfand widths $c_m(B_X,Y)$ we take an infimum over subspaces $M$ of $Y$.  In particular, if $X$ and $Y$ are quasi-Banach spaces then it is not clear whether the above Gelfand numbers and Gelfand widths are equivalent. Only recently, some sufficient conditions have been determined under which $c_m(\id:X\to Y)$ and $c_m(B_X,Y)$ are equivalent, see \cite{EdLa13}. Obviously, if $X$ and $Y$ are both $n$-dimensional, then the Gelfand numbers and Gelfand widths coincide.

{\em (ii)} A relation between Gelfand widths and a corresponding notion of minimal worst case errors, which are nowadays sometimes called {\it compressive $m$-widths} 
\cite[Def.\ 10.2, Thm.\ 10.4]{FoR13}, was first given in \cite[Thm.\ 5.4.1]{TrWaWo88}. These associated minimal worst-case errors are defined, for a given bounded set $K$ in a (quasi-)Banach space $Y$, 
as
\begin{equation} \label{eqn:compressWidth}
 E_m(K,Y) := \inf_{S_m} \sup_{f\in K} \|f - S_m(f)\|_Y
\end{equation}
where the infimum is taken over all maps $S_m = \varphi_m \circ N_m$, with $N_m: Y \to \R^m$ any linear and continuous measurement map and $\varphi_m: \R^m \to Y$ an arbitrary  recovery map. Note that the worst-case reconstruction error $E_m(X,Y)$ in \eqref{minwce} and $E_m(B_X,Y)$ in \eqref{eqn:compressWidth} are defined in a subtly 
different manner: in the first case we consider measurement maps $N_m: X \to \R^m$, whereas in the second case we consider maps $N_m: Y \to \R^m$. The compressive $m$-width and the $m$-th Gelfand width are equivalent under mild conditions \cite[Thm.\ 10.4]{FoR13}. In particular, if $\|\cdot\|_X$ and $\|\cdot\|_Y$ are quasi-norms on $\R^n$, then
\begin{equation}\label{Em_cm}
 \alpha_{Y}^{-1}c_m(B_X,Y) \leq E_m(B_X,Y) \leq 2\alpha_X c_m(B_X,Y).
\end{equation}
\end{remark}
In conclusion, the $m$-th Gelfand number $c_m(\id:X\to Y)$, the $m$-th Gelfand width $c_m(B_X,Y)$ and their associated minimal worst-case errors $E_m(X,Y)$ and $E_m(B_X,Y)$, respectively, can in general be four non-equivalent quantities. However, if $X$ and $Y$ are finite-dimensional quasi-normed spaces of the same dimension, then these quantities are equivalent up to constants depending only the quasi-norm constants of $X$ and $Y$. In particular, if $X=\ell_p^b(\ell_q^d)$ and $Y=\ell_r^b(\ell_u^d)$, then the constants in the equivalences depend only on $p,q,r$ and $u$. We will use these facts heavily in the sections that follow.  

\subsection{Generalized notions of sparsity and compressibility}
\label{sec:sparsity}

We fix some terminology concerning structured sparsity. Consider a block vector $x=(x_i)_{i=1}^b$ with $x_i\in \R^{d_i}$. As usual, we say that $x$ is \emph{$s$-sparse} if it has at most $s$ nonzero entries. We say $x$ is \emph{$s$-block sparse} if at most $s$ blocks of $x$ are non-zero. This type of sparsity occurs frequently in the signal processing literature (see e.g.\ \cite{ADR14,Eldar08,EKB10,ElM09}). Block sparsity is strongly related to (and can in fact be viewed as a generalization of) the \emph{joint sparsity} model, where one considers a signal consisting of several `channels' (for instance, the three color channels of an RGB image) and assumes that nonzero coefficients appear at the same location within each of the channels (see e.g.\ \cite{ElM09,fora08,grrascva08} and \cite[Section 1.2]{ADR14}). Block sparsity also plays an important role in high-dimensional statistics (see e.g.\ \cite[Chapters 4 and 8]{BuGe11}, \cite[Sections 2.1 and 3.2.1]{Wai14} and the references therein) and appears in machine learning, for instance in multi-task learning and multiple kernel learning (see e.g.\ \cite[Sections 1.3 and 1.5]{BJMO12} and the references therein). In this literature, it is sometimes of interest to consider \emph{group sparsity}, a generalization of the block sparsity model where the signal consists of a small number of groups which are allowed to overlap. There is a well-developed theory available on sufficient conditions for recovery of block and, more generally, group sparse vectors via the group Lasso, see e.g.\ \cite{BuGe11,HuZ10,LPGT11,NRWY12}. As a third form of sparsity, we consider sparsity-in-levels, which has emerged recently in the signal processing literature. A block vector $x=(x_i)_{i=1}^b$  is called \emph{$(t_i)_{i=1}^b$-sparse-in-levels} if each block $x_i$ is a $t_i$-sparse vector. Empirically it has been observed that certain signals exhibit this type of structured sparsity after expansion in an appropriate basis. For instance, the wavelet coefficients of natural images often exhibit an (approximately) sparse-in-levels structure, where the levels correspond to the wavelet scales. We refer to e.g.\ \cite{AHP14,BaH14,BBW17,LiA16} for recent work on sufficient recovery guarantees for sparse-in-levels signals.\par
To keep our exposition non-technical and to ensure that our sparsity concepts are linked to classical 
$\ell_p(\ell_q)$-spaces, we restrict ourselves to the study of block sparsity and sparsity-in-levels in the case where 
the block sizes and (in the case of sparsity-in-levels) the sparsity of the blocks are identical. 
Accordingly, we say that a block vector $x=(x_i)_{i=1}^b \in \R^{bd}$ with $x_i\in \R^d$ for all $i\in [b]$ is \emph{$s$-outer sparse} if at most $s$ blocks of $x$ are non-zero and say that $x$ is \emph{$t$-inner sparse} if each block of $x$ has at most $t$ non-zero entries. The qualifiers `outer'  and `inner' refer to the fact that the sparsity is defined with respect to the outer $\ell_p$-space or inner $\ell_q$-space, respectively, if $x\in \ell_p^b(\ell_q^d)$. For our analysis it will also be important to consider a mixture of outer and inner sparsity: we define $x$ to be \emph{$(s,t)$-sparse} if it has at most $s$ non-zero blocks, each of which has at most $t$-non-zero entries.\par
To illustrate why outer and inner sparsity should play an important role for the Gelfand numbers of $\ell_p^b(\ell_q^d)$-spaces where $p$ and/or $q$ are small, let us point out a generalization of  ``Stechkin's lemma'' 
\cite[Prop.\ 2.3, Thm.\ 2.5]{FoR13} to the two different best $s$-term approximation errors \eqref{eqn:beststermOI} associated with inner and outer sparsity. Let us emphasize that although \cite[Prop.\ 2.3, Thm.\ 2.5]{FoR13} is often attributed to S.B. Stechkin, he never formulated such a result. The first known version goes back to Temlyakov 1986 \cite[p.\ 92]{Te86} for $p,q\geq 1$. Note, that there is indeed a well-known Stechkin lemma (or Stechkin criterion) from 1955 proved in \cite{St55}, see also \cite[7.4]{DTU16} for further historical comments. 
This criterion implies the inequality \cite[Prop.\ 2.3]{FoR13} in case $p=1$ and $q=2$ but with a constant larger than $1$.  Let us now formulate two generalizations of \cite[Prop.\ 2.3]{FoR13}. Their proofs are straightforward modifications. If $0<p<r \leq \infty$, then for any $x\in \R^{bd}$,
\begin{equation}
\label{eqn:StechOuter}
\si_s^{\operatorname{outer}}(x)_{\ell_r^b(\ell_q^d)} \leq \frac{1}{s^{1/p-1/r}} \|x\|_{\ell_p^b(\ell_q^d)}
\end{equation}
and if $0<q<u \leq \infty$
\begin{equation}
\label{eqn:StechInner}
\si_t^{\operatorname{inner}}(x)_{\ell_p^b(\ell_u^d)} \leq \frac{1}{t^{1/q-1/u}} \|x\|_{\ell_p^b(\ell_q^d)}.
\end{equation}
The bound (\ref{eqn:StechOuter}) shows that vectors in $B_{\ell_p^b(\ell_q^d)}$ with $p$ small can be approximated well by outer-sparse 
vectors and, similarly, (\ref{eqn:StechInner}) implies that elements in $B_{\ell_p^b(\ell_q^d)}$ with $q$ small can be approximated well by inner-sparse vectors. The constant $c=1$ in 
(\ref{eqn:StechOuter}) and (\ref{eqn:StechInner}) is not optimal. It is actually below one, depending on $p,q,r,u$, see \cite[Section 7.4]{DTU16} and the references therein.

\section{A general upper bound for Gelfand widths of $K\subset \R^n$}

In this section we set up some machinery to prove upper bounds for Gelfand numbers based on ``Gordon's escape through a mesh'' theorem . To facilitate potential re-use, we write our results down in more generality than needed in Section~\ref{sec:matchGel} below. Below we will deal with Gelfand widths of bounded subsets in $\R^n$, see (\ref{eqn:GwidthDef}) in Remark \ref{rem:Gelfand_width}\,.

The main tool in the proof will be Gordon's ``Escape through a mesh'' theorem \cite[Cor.\ 3.4]{Gor88}. Recall that the Gaussian width of a set $T\subset \R^n$ is defined by
$$w(T) = \E\sup_{x\in T}\langle g,x\rangle,$$
where $g$ is an $n$-dimensional standard Gaussian vector. Let us denote 
$$
    E_n:=\mathbb{E} \|g\|_2 = \sqrt{2}\frac{\Gamma((n+1)/2)}{\Gamma(n/2)}\,.
$$
We will state the version of Gordon's theorem from \cite[Thm.\ 9.21]{FoR13}.
\begin{lemma}\label{Gor} Let $A \in \R^{m\times n}$ be a Gaussian random matrix and $T$ be a subset of the unit sphere $S^{n-1} := \{x\in \R^n~:~\|x\|_2 = 1 \}$. Then for every $t>0$ we have
$$
    \mathbb{P}\Big(\inf\limits_{x\in T} \|Ax\|_2 \leq E_m-w(T)-t\Big) \leq e^{-t^2/2}\,.
$$
\end{lemma}
By the known relation $m/\sqrt{m+1} \leq E_m \leq \sqrt{m}$ (see e.g.\  \cite[Prop.\ 8.1]{FoR13}), Lemma \ref{Gor} yields a non-trivial result if $E_m>w(T)$ and in particular if $m \gtrsim w(T)^2$\,.
\begin{proposition}
\label{prop:GGwidth}
Let $0<\rho<\infty$, $\|\cdot\|_Y$ denote a quasi-norm on $\R^n$ and $K\subset \R^n$ be a bounded symmetric set. Define the set 
$$K_{\rho} = \{x\in K \ : \ \|x\|_Y>\rho\}$$
and let $K_{\rho}^{\|\cdot\|_2} = \{x/\|x\|_2 \ : \ x\in K_{\rho}\}$ be the associated set of $\ell_2$-normalized vectors. If
\begin{equation}
\label{eqn:mGGwidth}
m\gtrsim \max\{w(K_{\rho}^{\|\cdot\|_2})^2,1\},
\end{equation}
then 
$$c_m(K,Y) \leq \rho.$$
\end{proposition}
\begin{proof}
To prove the assertion, we need to find an $A\in \R^{m\ti n}$ such that $\|x\|_Y\leq \rho$ for all $x\in \mathrm{ker}(A)\cap K$. 
This is satisfied if $\inf_{x\in K_{\rho}}\|A x\|_2^2>0$. 
By Lemma \ref{Gor}, if $A$ is an $m\ti n$ standard Gaussian matrix such that $E_m>w(K_{\rho}^{\|\cdot\|_2})$, then 
\begin{equation}\label{eq_Ax}
    \gamma(K,\rho;A):=\inf\limits_{x\in K_{\rho}^{\|\cdot\|_2}} \|Ax\|_2 > 0
\end{equation}
with probability at least $1-\exp(-t^2/2)$, where $0<t<E_m-w(K_{\rho}^{\|\cdot\|_2})$\,.
If this event happens, we have for $x\in K_{\rho}$ and $y = x/\|x\|_2$ the bound
$$
    \|Ax\|_2 = \|x\|_2\cdot \|A(y)\|_2 \geq C_Y\gamma(K,\rho;A)\|x\|_Y \geq C_Y\gamma(K,\rho;A)\rho>0\,,
$$
for a certain $C_Y>0$, as all quasi-norms on $\R^n$ are equivalent. Hence, \eqref{eqn:mGGwidth} ensures the existence of 
$A\in\R^{m\ti n}$ satisfying $\inf_{x\in K_{\rho}}\|A x\|_2>0$. 
\end{proof}
In conclusion, we can obtain an upper bound for the Gelfand width by estimating 
$$w(K_{\rho}^{\|\cdot\|_2}) = \E\sup_{x\in K \ : \ \|x\|_Y>\rho} \Big\langle \frac{x}{\|x\|_2},g\Big\rangle$$
from above.

\section{A lower bound for Gelfand widths via packing numbers}
\label{sec:lower_bound1}
We generalize the proof-by-contradiction technique from \cite{FPR10}, which was used to prove lower bounds for 
$c_m(\id:\ell_p^n \to \ell_q^n)$. To this end, we establish a connection between Gelfand and packing numbers in 
Proposition \ref{prop:GelEstPack}, which requires some further preparation. 
Recall that for a quasi-normed space $X$, a bounded subset $K \subset X$, and $\varepsilon>0$, the packing number $\mathcal P(K, \|\cdot\|_X, \varepsilon)$ 
is the maximal number of elements in $K$ which have pairwise distances strictly larger than $\varepsilon$. The following well-known lemma can be proven by a standard volume comparison argument.
\begin{lemma}
\label{lem:volBound}
Let $\|\cdot\|_X$ be any quasi-norm on $\R^n$ with quasi-norm constant $1\leq \al_X<\infty$ and let $U$ be a subset of the unit ball $B_X=\{x\in \R^n: \|x\|_X \leq 1\}$. Then for any $\varepsilon>0$,
\begin{equation*}
\mathcal{P}(U, \|\cdot\|_X,\varepsilon) \leq \al_X^n \left(1+\frac{2}{\varepsilon}\right)^n.
\end{equation*}
\end{lemma}
\begin{lemma}
\label{lem:minImpliesPacking}
Let $A\in\R^{m\ti n}$ and let $\|\cdot\|_X$ be any quasi-norm on $\R^n$. Suppose that $T\subset \R^n$ is a set of vectors satisfying, for some $c\geq 1$,
\begin{equation}
\label{eqn:assumpT}
\inf_{z\in \R^n \ : \ A z=A x} \|z\|_X \geq c^{-1}\|x\|_X \qquad \mathrm{for \ all} \ x\in T-T.
\end{equation}
Let $U\subset T$ and let $r_X(U) = \sup_{x\in U} \|x\|_X$ be its radius. Then, for any $\eps>0$,
$$\log\cP(U,\|\cdot\|_X,\eps) \leq m\log\Big(\al_X+\frac{2c\al_Xr_X(U)}{\eps}\Big),$$
where $\al_X$ is the quasi-norm constant.
\end{lemma}
\begin{proof}
Consider the quotient space $Q=X/\mathrm{ker}(A)$ equipped with its natural quasi-norm
$$\|[x]\|_Q : = \inf_{v \in \mathrm{ker}(A)} \|x-v\|_X.$$
Note that for every $v\in \mathrm{ker}(A)$, the vector $z=x-v$ satisfies $A z=A x$. On the other hand, if $z\in \R^n$ satisfies $A z=A x$, then $v=x-z$ defines a vector in $\mathrm{ker}(A)$. By (\ref{eqn:assumpT}), this implies  
$$\|[x]\|_Q \leq \|x\|_X \leq c\|[x]\|_Q \qquad \mathrm{for \ all} \ x\in T-T$$
and in particular
$$\cP(U,\|\cdot\|_X,\eps) \leq \cP(U,c\|\cdot\|_Q,\eps) = \cP(U,\|\cdot\|_Q,\eps/c).$$
Now apply Lemma~\ref{lem:volBound} to obtain
\begin{align*}
\cP(U,\|\cdot\|_X,\eps) & \leq \cP(r_X(U) B_{Q},\|\cdot\|_Q,\eps/c) \\
& \leq \Big(\al_X + \frac{2c\al_X r_X(U)}{\eps}\Big)^{\mathrm{rank}(A)} \leq \Big(\al_X + \frac{2c\al_X r_X(U)}{\eps}\Big)^m\,.
\end{align*}
Taking logarithms yields the result.
\end{proof}
In what follows, we let $\|\cdot\|_X$ be a quasi-norm on $\R^{bd}$ with quasi-norm constant $1\leq \al_X<\infty$. Moreover, we let $0<\beta_X<\infty$ denote the smallest constant such that for all $S\subset [b]\ti[d]$ and all $x\in \R^{bd}$ we have
$$\|x_S\|_X + \|x_{S^c}\|_X\leq \beta_X \|x\|_X.$$
Clearly, for any $S\subset [b]\ti[d]$ and $x\in \R^{bd}$,
$$\|x\|_X\leq \al_X(\|x_S\|_X + \|x_{S^c}\|_X)\leq \al_X\beta_X \|x\|_X,$$
so $\al_X\beta_X\geq 1$. Note that $\beta_{\ell_p^b(\ell_q^d)}\leq 2$ depends only on $p$ and $q$.
\begin{lemma}
\label{lem:NSPNew}
Let $\cS$ be a collection of supports in $[b]\ti[d]$ and let $A\in \R^{m\times bd}$. Suppose that for every $v\in \mathrm{Ker}(A)\setminus\{0\}$ and $S\in \cS$ we have $\|v_S\|_{X}<\|v_{S^c}\|_{X}$. Then, for any $x\in \R^{bd}$ and $S\in \cS$,
$$\inf_{z \in \R^{bd} \ : \ A z = A x_S} \|z\|_X > (\al_X\beta_X)^{-1}\|x_S\|_X.$$
\end{lemma}
\begin{proof}
Let $x \in \R^{bd}$. If $z\neq x_S$ satisfies $A z=A x_S$, then $v=x_S-z \in \mathrm{Ker}(A)\setminus\{0\}$ and therefore
\begin{align*}
\al_X^{-1}\|x_S\|_X & \leq \|x_S-z_S\|_X + \|z_S\|_X = \|v_S\|_X + \|z_S\|_X \\
& < \|v_{S^c}\|_X + \|z_S\|_X \\
& = \|z_{S^c}\|_X + \|z_S\|_X \leq \beta_X\|z\|_X.
\end{align*}
\end{proof}
Let us define the parameter 
$$\ka_{X,Y,\cS} = \sup_{x \in \R^{bd},S\in \cS} \frac{\|x_S\|_X}{\|x_S\|_Y}.$$
\begin{proposition}
\label{prop:GelEstPack} 
Let $\cS$ be a collection of supports in $[b]\ti [d]$ and define 
$\cS_+ = \{S_1\cup S_2 \ : \ S_1,S_2 \in \cS\}.$
Let $\|\cdot\|_Y$ be a quasi-norm on $\R^{bd}$ satisfying 
$$\|x_S\|_Y \leq \|x\|_Y \qquad \mathrm{for \ all} \ x\in \R^n, S\in \cS_+.$$
If
$$c_m(B_{X}, Y) < (2\al_X\ka_{X,Y,\cS_+})^{-1}$$
then for any $U\subset \{x_S \ : \ x\in \R^{bd}, \ S \in \cS\}$ and $\eps>0$,
$$\log\cP(U, \|\cdot\|_X,\eps) \leq m\log\Big(\al_X+\frac{2\al_X^2\beta_X r_X(U)}{\eps}\Big).$$
\end{proposition}
\begin{proof}
By definition of the Gelfand width, there exists an $A\in\R^{m\ti bd}$ such that for every $x\in \mathrm{Ker}(A)\setminus\{0\}$
$$\|x\|_Y<(2\al_X\ka_{X,Y,\cS_+})^{-1}\|x\|_X.$$
For any $S\in \cS_+$, 
\begin{align*}
\|x_S\|_{X} & \leq \ka_{X,Y,\cS_+} \|x_S\|_Y \\
& \leq \ka_{X,Y,\cS_+} \|x\|_Y < (2\al_X)^{-1}\|x\|_{X} \leq \frac{1}{2}(\|x_S\|_X + \|x_{S^c}\|_X). 
\end{align*}
Rearranging, we conclude that $\|x_S\|_{X}<\|x_{S^c}\|_X$. Now apply Lemma~\ref{lem:NSPNew} and subsequently Lemma~\ref{lem:minImpliesPacking} (with $T=\{x_S \ : \ x\in \R^{bd}, \ S \in \cS\}$, noting that $T-T=\{x_S \ : \ x\in \R^{bd}, \ S \in \cS_+\}$) to obtain the conclusion. 
\end{proof}
Proposition~\ref{prop:GelEstPack} can be used to prove lower bounds for the Gelfand widths $c_m(B_{\ell_p^b(\ell_q^d)}, \ell_r^b(\ell_u^d))$ via contradiction: 
one assumes that $c_m(B_{\ell_p^b(\ell_q^d)}, \ell_r^b(\ell_u^d))$ is small and then constructs a large packing consisting of structured 
sparse vectors to find a contradiction. This packing construction is the subject of the next section.

\section{Construction of a large packing in $B_{\ell_p^b(\ell_q^d)}$}
\label{sec:lower_bound2}
We obtain a lower bound for the packing numbers $\mathcal P(B_{\ell_p^b(\ell_q^d)}, \|\cdot\|_{\ell_r^b(\ell_u^d)}, \varepsilon)$, 
where $0<p,q,r,u \leq \infty$. The proof, inspired by \cite{ADR14}, 
is based on two combinatorial facts. The first combinatorial fact has been independently observed in various disciplines of mathematics. A proof can be found, e.g.,\ in \cite[Lemma 10.12]{FoR13}, \cite{Kuehn2001} or 
\cite[Prop.\ 2.21, Page 219]{Pin85}\,.
\begin{lemma}\label{lem:subsets}
Given integers $\ell < n$, there exist subsets $I_1, \ldots, I_N$ of $[n]$ such that each $I_i$ has cardinality $2\ell$ and
$$ \card(I_i \cap I_j) < \ell \ \ \text{ whenever } i \neq j,$$
and 
\begin{align*}
N \geq \Big(\frac{n}{8\ell}\Big)^{\ell}.
\end{align*}
\end{lemma}
\noindent The second combinatorial fact is a variation of Lemma \ref{lem:subsets}, which is well known in coding theory \cite{Gilbert52,Varshamov57}. For given sets $A_1,\ldots,A_\ell$ we let $d_H$ denote the Hamming distance on $A_1\ti\cdots\ti A_\ell$, i.e.,
$$d_H((x_1,\ldots,x_{\ell}),(y_1,\ldots,y_{\ell})) = \sum_{i=1}^{\ell} 1_{x_i\neq y_i}.$$
\begin{lemma}\label{lem:Gilbert}(Gilbert-Varshamov bound) Fix $\theta>1$. Let $A_1,\ldots,A_\ell$ be sets, each consisting of $\theta$ elements. Then, for any $k\in [\ell]$,
\begin{align*}
\mathcal{P}(A_1\ti\cdots\ti A_\ell,d_H,k) \geq \frac{\theta^\ell}{\sum_{j=0}^{k-1} \binom{\ell}{j} (\theta-1)^j }.
\end{align*}
\end{lemma}
Let us recall that we defined a block vector $x = (x_i)_{i=1}^b \in \R^{bd}$ to be $(s,t)$-sparse if it has at most $s$ non-zero blocks, each of which has at most $t$ non-zero entries. In the proof of the following proposition it will be notationally convenient to identify each $x\in \R^{bd}$ with the $b\times d$ matrix having the $x_i$ as its rows. Under this identification, the matrix $x$ is $(s,t)$-sparse if it has at most $s$ non-zero rows, each of which has at most $t$ non-zero entries.
\begin{proposition}
\label{prop:sparsePacking}
Let $b,d \in \N$ with $b,d \geq 8$. For any $0<p,q,r,u\leq \infty$, $s \in [\lfloor b/8\rfloor]$, and $t \in [\lfloor d/8 \rfloor]$, there is a set $W$ of $(2s,2t)$-sparse vectors with $r_{\ell_p^b(\ell_q^d)}(W)\leq (2s)^{1/p}(2t)^{1/q}$ such that
$$\cP(W,\|\cdot\|_{\ell_r^b(\ell_u^d)},s^{1/r} (2t)^{1/u})\geq \Big(\frac{b}{32s}\Big)^s\Big(\frac{d}{8t}\Big)^{st}.$$
\end{proposition}
\begin{proof}
We use Lemma~\ref{lem:subsets} to find a collection $A_{d,t}$ of $(d/(8 t))^{t}$ subsets of $[d]$ such that each $I\in A_{d,t}$ has cardinality $2t$ and
$$\card(I \cap J) < t \ \ \text{ whenever } I \neq J.$$
In particular, the symmetric difference $I\Delta J$ has cardinality at least $2t$ if $I\neq J$. Set $A_i:=A_{d,t}$ for $i\in [2s]$. Applying Lemma~\ref{lem:Gilbert} to $A_1\times\cdots\times A_{2s}$ with $\ell=2s$, $\theta=(d/(8 t))^{t}$ and $k=s$, we find a set $\mathcal{A}\subset A_1\times\cdots\times A_{2s}$ with
\begin{align*}
\card(\mathcal{A}) \geq \frac{(d/(8 t))^{2st}}{\sum_{j=0}^{s-1} \binom{2s}{j} ((d/(8 t))^{t}-1)^j } \geq \frac{(d/(8 t))^{2st}}{2^{2s} (d/(8 t))^{st}} = \frac{(d/(8 t))^{st}}{4^{s}},
\end{align*}
so that for any $K=(K_1,\ldots,K_{2s})$ and $L=(L_1,\ldots,L_{2s})$ in $\mathcal{A}$ there are at least $s$ indices $i$ with $K_i\neq L_i$ and in particular $\card(K_i\cap L_i)<t$. We now define a set of $(2s,2t)$-sparse matrices in $\R^{2s\times d}$ with $\{0,1\}$-entries by
$$M=\Big\{\sum_{i=1}^{2s} \sum_{j\in K_i} e_{ij} \ : \ K\in \mathcal{A}\Big\}.$$ 
By construction, for each pair $w,z$ in $M$ with $w\neq z$ there exist at least $s$ indices $i$ so that 
\begin{equation}
\label{eqn:propertyM}
\text{supp}(w_i)\cap \text{supp}(z_i)<t.
\end{equation}
Next, using Lemma~\ref{lem:subsets} we pick a sequence $(I_{\al})$ of $(b/(8s))^s$ subsets of $[b]$ such that each $I_{\al}$ has cardinality $2s$ and
$$\card(I_{\al} \cap I_{\ga}) < s \ \ \text{ whenever } \al \neq \ga.$$
We now define $W$ to be the set of $(2s,2t)$-sparse matrices
$$W = \bigcup_{\alpha}  \Big\{x \in \R^{b\times d}: x_{I_\alpha\times[d]} \in M, x_{([b] \setminus I_\alpha)\times [d]} = 0 \Big\}.$$
Observe that 
$$\card(W) = (b/(8s))^s \card(M)\geq (b/(8s))^s \frac{(d/(8 t))^{st}}{4^{s}} = \Big(\frac{b}{32s}\Big)^s\Big(\frac{d}{8t}\Big)^{st}.$$ 
For $I_{\alpha}$ and $K\in \mathcal{A}$ let us denote
$$x_{I_{\alpha},K} = \sum_{i\in I_{\alpha}} \sum_{j\in K_i} e_{ij}.$$
By construction, 
$$\|x_{I_{\alpha},K}\|_{\ell_p^b(\ell_q^d)} = \Big(\sum_{i\in I_{\al}}\Big(\sum_{j\in K_i}\Big)^{p/q}\Big)^{1/q}\leq (2s)^{1/p}(2t)^{1/q},$$
so $r_{\ell_p^b(\ell_q^d)}(W)\leq (2s)^{1/p}(2t)^{1/q}$. Moreover, for $K,L\in \mathcal{A}$ with $K\neq L$, the property \eqref{eqn:propertyM} implies
$$\|x_{I_{\alpha},K} - x_{I_{\alpha},L}\|_{\ell_r^b(\ell_u^d)} = \Big(\sum_{i\in I_{\al}} \Big(\sum_{j\in K_i \Delta L_i}\Big)^{r/u}\Big)^{1/r} \geq s^{1/r} (2t)^{1/u}.$$
Finally, if $I_{\alpha}\neq I_{\ga}$, then   
\begin{align*}
\|x_{I_{\alpha},K} - x_{I_{\ga},L}\|_{\ell_r^b(\ell_u^d)} & = \Big(\sum_{i\in I_{\al}\cap I_{\ga}} \Big(\sum_{j\in K_i \Delta L_i}\Big)^{r/u} + \sum_{i\in I_{\al}\setminus I_{\ga}} \Big(\sum_{j\in K_i}\Big)^{r/u} \\
& \qquad \qquad + \sum_{i\in I_{\ga}\setminus I_{\al}} \Big(\sum_{j\in L_i}\Big)^{r/u}\Big)^{1/r} \geq (2s)^{1/r} (2t)^{1/u}.
\end{align*}
In conclusion, $W$ is an $s^{1/r} (2t)^{1/u}$-packing. 
\end{proof}

\section{Matching bounds for Gelfand numbers}
\label{sec:matchGel}

This section is devoted to the proof of the following extension of Theorem~\ref{thm:GelSimple}.
\begin{theorem}
\label{thm:GelMain}
Let $m\leq bd$. Then the following assertions hold:
\begin{itemize}
\item[(i)] If $0<p\leq 1$ and $p<q\leq 2$, then 
\begin{equation}
\label{eqn:GelMain1}
c_m(\id:\ell_p^b(\ell_2^d) \to \ell_q^b(\ell_2^d)) \simeq_{p,q} \min\Big\{1,\frac{\log(eb/m)+d}{m}\Big\}^{1/p-1/q}
\end{equation}
and in particular,
\begin{equation}
\label{eqn:Foucartetal}
c_m(\id:\ell_p^b(\ell_p^d) \to \ell_q^b(\ell_q^d)) \simeq_{p,q} \min\Big\{1,\frac{\log(ebd/m)}{m}\Big\}^{1/p-1/q}\,.
\end{equation}
\item[(ii)] If $0<q\leq 1$ and $q\leq p\leq 2$, then 
$$c_m(\id:\ell_p^b(\ell_q^d) \to \ell_p^b(\ell_p^d)) \simeq_{p,q} \min\Big\{1,\frac{b\log(ebd/m)}{m}\Big\}^{1/q-1/p}.$$
\item[(iii)] Set $0<q\leq 1$ and $q\leq p\leq 1$. Then, there is a constant $c_{p,q}$, so that
\begin{align*}
 c_m(\id:\ell_p^b(\ell_q^d) \to \ell_2^b(\ell_2^d)) \simeq_{p,q}
 \begin{cases}
  1 &: 1 \leq m \leq c_{p,q}\log\Big(\frac{ebd}{m}\Big)\,,\\
  \left(\frac{\log(ebd/m)}{m}\right)^{1/p-1/2} &: c_{p,q}\log\Big(\frac{ebd}{m}\Big) \leq m \leq c_{p,q}b\log\Big(\frac{ebd}{m}\Big)\,,\\
  b^{1/2-1/p} \left(\frac{b\log(ebd/m)}{m}\right)^{1/q-1/2} &: c_{p,q}b\log\Big(\frac{ebd}{m}\Big) \leq m \leq bd\,.
 \end{cases}
\end{align*}
\end{itemize}
\end{theorem}
Clearly, (\ref{eqn:Foucartetal}) immediately follows from (\ref{eqn:GelMain1}) by taking $d=1$ and `$b=bd$'. 
The bounds in (\ref{eqn:Foucartetal}) were obtained earlier in \cite{FPR10,Vyb08}.
\begin{remark}
The two-sided estimate in (iii) in the case $c_{p,q}b\log\Big(\frac{ebd}{m}\Big) \leq m \leq bd$ remains valid if $0<q\leq 1$ and $q\leq p\leq 2$. This will be clear from the proof.
\end{remark} 
\begin{remark}
Our proof of the upper bounds in Theorem~\ref{thm:GelMain} is non-constructive. As a matter of fact, our proof shows (see in particular the proof of Proposition~\ref{prop:GGwidth}) 
that the two-sided bounds in Theorem~\ref{thm:GelMain} are with high probability attained by the kernel of an $m\times bd$ standard Gaussian matrix. 
\end{remark} 
We start by proving the upper bound in (\ref{eqn:GelMain1}). Consider the set 
$$D_{b,d,s} = \{x\in \R^{bd} \ : \ \|x\|_{\ell_1^b(\ell_2^d)}\leq \sqrt{s}, \|x\|_{\ell_2^b(\ell_2^d)}\leq 1\}.$$
We will need the following observation, which is known in the case $d=1$ (see \cite[Lemma 3.1]{PlV13a} and \cite[Lemma 2.9]{PlV13b}).
\begin{lemma}
\label{lem:CHsparseBlock}
Let $L_{b,d,s}$ be the set of all $s$-outer sparse vectors in $B_{\ell_2^b(\ell_2^d)}$. Then, 
\begin{equation}
\label{eqn:CHsparseBlock}
\mathrm{conv}(L_{b,d,s})\subset D_{b,d,s} \subset 2\mathrm{conv}(L_{b,d,s}).
\end{equation}
As a consequence,
\begin{equation}
\label{eqn:wDnsEst}
w(D_{b,d,s}) \lesssim \sqrt{s\log(eb/s)} + \sqrt{sd}.
\end{equation}
\end{lemma}
\begin{proof}
The first inclusion in (\ref{eqn:CHsparseBlock}) follows immediately, since by Cauchy-Schwarz, $\|x\|_{\ell_1^b(\ell_2^d)}\leq \sqrt{s}$ if $x\in L_{b,d,s}$. 
For the second inclusion, let $x\in D_{b,d,s}$. We write $x=\sum_{k\geq 0}x_{S_k\ti[d]}$, where the $S_k$ are disjoint subsets of $[b]$, 
satisfying $|S_k|=s$ and $\|x_{i\cdot}\|_{\ell_2^d}\geq \|x_{j\cdot}\|_{\ell_2^d}$ if $i\in S_{k-1}$ and $j \in S_k$. In particular, for $k\geq 1$,
$$\|x_{S_k\ti[d]}\|_{\ell_2^b(\ell_2^d)}\leq \sqrt{s} \max_{i\in S_k} \|x_{i\cdot}\|_{\ell_2^d} \leq \sqrt{s} \min_{j\in S_{k-1}} \|x_{j\cdot}\|_{\ell_2^d} \leq \frac{1}{\sqrt{s}}\|x_{S_{k-1}\ti[d]}\|_{\ell_1^b(\ell_2^d)}.$$
Note, moreover, that $\|x_{S_0\ti[d]}\|_{\ell_2^b(\ell_2^d)}\leq \|x\|_{\ell_2^b(\ell_2^d)}\leq 1$ and therefore
$$\sum_{k\geq 0}\|x_{S_k\ti[d]}\|_{\ell_2^b(\ell_2^d)} \leq 1 + \sum_{k\geq 1}\|x_{S_k\ti[d]}\|_{\ell_2^b(\ell_2^d)} \leq 1 + \sum_{k\geq 0}\frac{1}{\sqrt{s}}\|x_{S_k\ti[d]}\|_{\ell_1^b(\ell_2^d)}
= 1+\frac{1}{\sqrt{s}}\|x\|_{\ell_1^b(\ell_2^d)} \leq 2.$$
Writing
$$x = \sum_{k\geq 0} \|x_{S_k\ti[d]}\|_{\ell_2^b(\ell_2^d)}\frac{x_{S_k\ti[d]}}{\|x_{S_k\ti[d]}\|_{\ell_2^b(\ell_2^d)}}$$
and noting that $x_{S_k\ti[d]}/\|x_{S_k\ti[d]}\|_{\ell_2^b(\ell_2^d)}$ is in $L_{b,d,s}$, we conclude that the second inclusion in (\ref{eqn:CHsparseBlock}) holds. This immediately implies that 
$$w(D_{b,d,s}) \leq 2\E\sup_{y \in L_{b,d,s}} \langle y,g\rangle = 2\E\sup_{S\in \cS} \|g_S\|_{\ell_2^b(\ell_2^d)} = 2\E\sup_{S\in \cS} (\|g_S\|_{\ell_2^b(\ell_2^d)} - \E \|g_S\|_{\ell_2^b(\ell_2^d)}) +  2\sup_{S\in \cS}\E\|g_S\|_{\ell_2^b(\ell_2^d)}, $$
where $g$ is a standard Gaussian vector of length $bd$ and 
$$\cS = \{\tilde{S}\ti[d] \ : \ \tilde{S}\subset [b], \ \card(\tilde{S})=s\}$$ 
is the set of all $s$-outer sparse support sets in $[b]\ti [d]$. Clearly,
$$\sup_{S\in \cS}\E\|g_S\|_{\ell_2^b(\ell_2^d)} \leq \sup_{S\in \cS}(\E\|g_S\|_{\ell_2^b(\ell_2^d)}^2)^{1/2} = \sqrt{sd}.$$
By the Gaussian concentration inequality for Lipschitz functions (see e.g.\ \cite[Theorem 5.6]{BLM13}), for any $S\in \cS$,
$$\bP(|\|g_S\|_{\ell_2^b(\ell_2^d)} - \E\|g_S\|_{\ell_2^b(\ell_2^d)}|\geq t)\leq 2e^{-t^2/2} \qquad (t>0).$$
Thus, the $\|g_S\|_{\ell_2^b(\ell_2^d)} - \E\|g_S\|_{\ell_2^b(\ell_2^d)}$ are mean-zero and subgaussian and therefore (see e.g.\ \cite[Proposition 7.29]{FoR13}),
$$\E\sup_{S\in \cS} (\|g_S\|_{\ell_2^b(\ell_2^d)} - \E \|g_S\|_{\ell_2^b(\ell_2^d)}) \lesssim (\log\card(\cS))^{1/2} = \Big(\log {b \choose s}\Big)^{1/2} \leq \sqrt{s\log(eb/s)}.$$
Combining these estimates yields (\ref{eqn:wDnsEst}).
\end{proof}
We will also use the following technical lemma. It formalizes, in a very special case, the idea that the functions $x \mapsto x/\log(1/x)$ and $y \mapsto y \log(1/y)$ can be treated as `inverse functions up to constants'. 
\begin{lemma}
\label{lem:invert}
Let $C\geq 1$, $x>0$, $0<y\leq K$ and 
$$x\leq \frac{1}{Ce} \frac{y}{\log(eK/y)}.$$
Then,
$$y\geq \frac{Ce}{1+\log(Ce)} x\log(eK/x).$$
\end{lemma}
\begin{proof}
As $y\leq K$, we find 
$$\frac{x}{y}\leq \frac{1}{Ce} \frac{1}{\log(eK/y)} \leq \frac{1}{Ce}.$$
Since $t\mapsto t\log(t)$ is decreasing on $[0,1/e]$, we obtain
\begin{align*}
y\geq Ce x\log(eK/y) & = Ce x\log(eK/x)+Ce y\frac{x}{y}\log(x/y) \\
& \geq Ce x\log(eK/x)+Ce y\frac{1}{Ce}\log(1/(Ce)) = Ce x\log(eK/x) - y\log(Ce).
\end{align*}
Rearranging yields the asserted bound.
\end{proof}
\begin{proposition}
Set $0<p\leq 1$ and $p<q\leq 2$. There is an absolute constant $C\geq 1$ such that for all $m\leq bd$ we have
$$c_m(\id:\ell_p^b(\ell_2^d) \to \ell_q^b(\ell_2^d)) \leq \min\Big\{1,Ce\frac{\log(eb/m)+d}{m}\Big\}^{1/p-1/q}.$$
\end{proposition}
\begin{proof}
If $m\leq Ce(\log(eb/m)+d)$ then we can trivially bound 
$$c_m(\id:\ell_p^b(\ell_2^d) \to \ell_q^b(\ell_2^d))\leq \|\id:\ell_p^b(\ell_2^d) \to \ell_p^b(\ell_2^d)\| = 1,$$
which is the desired estimate.\par
Assume now that $m>Ce(\log(eb/m)+d)$. We apply Proposition~\ref{prop:GGwidth} for 
$K=B_{\ell_p^b(\ell_2^d)}$ and $Y=\ell_q^b(\ell_2^d)$ to deduce the result. 
To this end, we determine for a given $s \in \N$ a suitable $0 < \rho < 1$ such that $K_{\rho}^{\|\cdot \|_2}\subset D_{b,d,s}$. 
This will allow for an upper estimate of the Gaussian width $w(K_{\rho}^{\|\cdot \|_2})$ in terms of $s$ based on Lemma~\ref{lem:CHsparseBlock}. 
Afterwards, we choose $s$ appropriately for the given $m$. Let $0\leq \theta\leq 1$ satisfy
$1/q = \theta/p + (1-\theta)/2$,
that is
$$\theta = \Big(\frac{1}{q}-\frac{1}{2}\Big)\Big(\frac{1}{p}-\frac{1}{2}\Big)^{-1}.$$
If $x\in B_{\ell_p^b(\ell_2^d)}$, then for any $x\in K_{\rho}$ we have
$$\rho<\|x\|_{\ell_q^b(\ell_2^d)}\leq \|x\|_{\ell_p^b(\ell_2^d)}^{\theta}\|x\|_{\ell_2^b(\ell_2^d)}^{1-\theta}\leq \|x\|_{\ell_2^b(\ell_2^d)}^{1-\theta}$$
or 
$$\|x\|_{\ell_2^b(\ell_2^d)}^{-1}\leq \rho^{-1/(1-\theta)}.$$
Let $0\leq \eta\leq 1$ satisfy 
$1 = \eta/p + (1-\eta)/2$,
that is,
$$\eta = \frac{1}{2}\Big(\frac{1}{p}-\frac{1}{2}\Big)^{-1}.$$
Then, for any $y=x/\|x\|_{\ell_2^b(\ell_2^d)} \in K_{\rho}^{\|\cdot \|_2}$ we have
$$\|y\|_{\ell_1^b(\ell_2^d)} = \frac{\|x\|_{\ell_1^b(\ell_2^d)}}{\|x\|_{\ell_2^b(\ell_2^d)}} \leq \frac{\|x\|_{\ell_p^b(\ell_2^d)}^{\eta}\|x\|_{\ell_2^b(\ell_2^d)}^{1-\eta}}{\|x\|_{\ell_2^b(\ell_2^d)}}\leq \|x\|_{\ell_2^b(\ell_2^d)}^{-\eta} \leq \rho^{-\eta/(1-\theta)}.$$
Note that 
$$\frac{\eta}{1-\theta} = \frac{1}{2}\Big(\frac{1}{p}-\frac{1}{q}\Big)^{-1}.$$
Therefore, if we set $\rho=s^{-(1/p-1/q)}$, then $K_{\rho}^{\|\cdot\|_2}\subset D_{b,d,s}$ and by Lemma~\ref{lem:CHsparseBlock},
$$w(K_{\rho}^{\|\cdot\|_2})\lesssim \sqrt{s\log(eb/s)} + \sqrt{sd}.$$
Proposition~\ref{prop:GGwidth} implies that
$$c_m(\id:\ell_p^b(\ell_2^d) \to \ell_q^b(\ell_2^d)) \leq s^{-(1/p-1/q)}$$
provided that $s$ satisfies
\begin{equation*}
m\gtrsim s\log(eb/s) + sd.
\end{equation*}
Set $C\geq 1$. Define $s=\lfloor m/(Ce(\log(eb/m)+d))\rfloor$, so that $s\in \N$ and 
$$s\leq \frac{m}{Ce(\log(eb/m)+d)} = \frac{m}{Ce\log(ebe^d/m)}.$$
Since $m\leq be^d$, we can now apply Lemma~\ref{lem:invert} with $K=be^d$ to obtain 
$$m\geq \frac{Ce}{1+\log(Ce)}s\log(ebe^d/s) = \frac{Ce}{1+\log(Ce)}(s\log(eb/s) + sd).$$
Thus, the desired condition on $m$ is satisfied if we pick $C\geq 1$ large enough. 
\end{proof}
We now prove the lower bound in (\ref{eqn:GelMain1}).
\begin{proposition}\label{res:GelLowerBound}
Set $0<p<q\leq 2$. There exist constants $c_{p,q}$ and $c_p$ such that for all $m\leq bd$ we have
$$c_m(\id:\ell_p^b(\ell_2^d)\to \ell_q^b(\ell_2^d)) \geq c_{p,q}\min\Big\{1,\frac{\frac{1}{2}c_p(\log(b/m)+(8e)^{-1}d)}{m}\Big\}^{1/p-1/q}.$$ 
\end{proposition}
\begin{proof}
We prove the result by contradiction using Proposition~\ref{prop:GelEstPack}, so suppose that
$$c_m(\id:\ell_p^b(\ell_2^d)\to \ell_q^b(\ell_2^d))<c_{p,q}\mu(m)^{1/p-1/q},$$
where 
$$\mu(m):=\min\Big\{1,\frac{\frac{1}{2}c_p(\log(b/m)+(8e)^{-1}d)}{m}\Big\}$$
and $c_p$, $c_{p,q}$ are to be determined below. Set $s=\lfloor \mu(m)^{-1}\rfloor$. Note that $s\in \N$ and 
$$\frac{1}{2\mu(m)}<s\leq \frac{1}{\mu(m)}.$$
For $x\in \mathrm{ker}(A)$, $x\neq 0$, we have by assumption
$$\|x\|_{\ell_p^b(\ell_2^d)}\leq b^{1/p-1/q}\|x\|_{\ell_q^b(\ell_2^d)}\leq b^{1/p-1/q}c_{p,q}s^{-(1/p-1/q)}\|x\|_{\ell_p^b(\ell_2^d)}$$
so in particular,
\begin{equation}
\label{eqn:bovers}
b/s\geq c_{p,q}^{-1/(1/p-1/q)}.
\end{equation}
By taking $c_{p,q}$ small enough, we can ensure that $s\leq b/8$, so that we can apply Proposition~\ref{prop:sparsePacking} below.\par
We apply Proposition~\ref{prop:GelEstPack} for $X=\ell_p^b(\ell_2^d)$ and $Y=\ell_q^b(\ell_2^d)$. Let 
$$\cS = \{\tilde{S}\ti[d] \ : \ \tilde{S}\subset [b], \ \card(\tilde{S})=2s\}$$
be the set of all $2s$-outer sparse support sets, so that $\cS_+$ is the set of all $4s$-outer sparse support sets. By the sharp inequality
$$\Big(\sum_{k=1}^{4s} \|x_k\|_2^p\Big)^{1/p} \leq (4s)^{1/p-1/q}\Big(\sum_{k=1}^{4s} \|x_k\|_2^q\Big)^{1/q}\qquad \operatorname{for \ all} \ x_1,\ldots,x_{4s}\in \R^d$$
we see that
$$\ka_{X,Y,\cS_+} = (4s)^{1/p-1/q}.$$
Therefore, if 
$$c_{p,q}\leq \Big(\frac{1}{4}\Big)^{1/p-1/q}(2\al_{\ell_p^b(\ell_2^d)})^{-1},$$
then by assumption
$$c_m(\id:\ell_p^b(\ell_2^d)\to \ell_q^b(\ell_2^d))<c_{p,q}\mu(m)^{1/p-1/q}\leq (2\al_{\ell_p^b(\ell_2^d)})^{-1}(4s)^{-(1/p-1/q)}.$$
By Proposition~\ref{prop:GelEstPack}, for any $W\subset \{x_S \ : \ S \in \cS\}$ and any $\eps>0$,
$$\log\cP(W, \|\cdot\|_{\ell_p^b(\ell_2^d)},\eps) \leq m\log\Big(\al_{\ell_p^b(\ell_2^d)}+\frac{2\al_{\ell_p^b(\ell_2^d)}^2\beta_{\ell_p^b(\ell_2^d)}r_{\ell_p^b(\ell_2^d)}(W)}{\eps}\Big).$$
By Proposition~\ref{prop:sparsePacking}, applied with $t=d/(8e)$, there is a set $W$ of $2s$-outer sparse vectors with $r_{\ell_p^b(\ell_2^d)}(W)\leq (2s)^{1/p}(d/4e)^{1/2}$ and 
$$\cP\Big(W, \|\cdot\|_{\ell_p^b(\ell_2^d)},s^{1/p}\Big(\frac{d}{4e}\Big)^{1/2}\Big)\geq \Big(\frac{b}{32s}\Big)^s e^{sd/(8e)}.$$
This implies
\begin{equation}
\label{eqn:boversAppl}
s\log(b/32s) + (8e)^{-1}sd\leq m\log(\al_{\ell_p^b(\ell_2^d)}+2^{1+1/p}\al_{\ell_p^b(\ell_2^d)}^2\beta_{\ell_p^b(\ell_2^d)}).
\end{equation}
Setting
$$c_p=\Big(\log(\al_{\ell_p^b(\ell_2^d)}+2^{1+1/p}\al_{\ell_p^b(\ell_2^d)}^2\beta_{\ell_p^b(\ell_2^d)})\Big)^{-1},$$
it follows from (\ref{eqn:bovers}) that
$$m\geq c_p s\log(b/32s) \geq c_p s\log(c_{p,q}^{-(1/p-1/q)}/32).$$
Hence, by picking $c_{p,q}$ small enough, we can ensure that $m\geq 32s$. Using this in (\ref{eqn:boversAppl}) yields
$$s\log(b/m) + (8e)^{-1}sd\leq mc_p^{-1}.$$
We now rearrange to find
\begin{align*}
m\geq c_p(s\log(b/m) + (8e)^{-1}sd) & >\frac{1}{2}\mu(m)^{-1} c_p(\log(b/m)+(8e)^{-1}d) \\
& = \frac{1}{2}c_p(\log(b/m)+(8e)^{-1}d)\Big(\min\Big\{1,\frac{\frac{1}{2}c_p(\log(b/m)+(8e)^{-1}d)}{m}\Big\}\Big)^{-1} \\
& \geq m,
\end{align*}
which is the desired contradiction. 
\end{proof}
\begin{remark}
From the previous proof it is clear that for $0<p<q\leq 2$, any $0<r<\infty$ and for all $m\leq bd$ we have
$$c_m(\id:\ell_p^b(\ell_r^d)\to \ell_q^b(\ell_r^d)) \geq c_{p,q,r}\min\Big\{1,\frac{\frac{1}{2}c_{p,r}(\log(b/m)+(8e)^{-1}d)}{m}\Big\}^{1/p-1/q}.$$
\end{remark}
We will now prove statement (ii) in Theorem~\ref{thm:GelMain}. The upper bound will follow by simple inclusion arguments. For the proof of the lower bound, we follow the same line of argument as in the proof of Theorem \ref{res:GelLowerBound}. The only difference is that this time, we keep the outer sparsity $s$ fixed and vary the inner sparsity $t$.
\begin{proposition}
Let $0<q\leq 1$ and $q\leq p\leq 2$. For all $m\leq bd$ we have
$$c_m(\id:\ell_p^b(\ell_q^d) \to \ell_p^b(\ell_p^d)) \simeq_{p,q} \min\Big\{1,\Big(\frac{b\log(ebd/m)}{m}\Big)^{1/q-1/p}\Big\}.$$
\end{proposition}
\begin{proof}
\textbf{Upper bound}: First of all, we have the trivial bound
$$c_m(\id:\ell_p^b(\ell_q^d) \to \ell_p^b(\ell_p^d))\leq c_0(\id:\ell_p^b(\ell_q^d)\to \ell_p^b(\ell_q^d))\leq 1.$$
Moreover, since $\|\id:\ell_p^b(\ell_q^d) \to \ell_q^b(\ell_q^d)\| = b^{1/q-1/p}$, we find using (\ref{eqn:Foucartetal}) and (S3)
$$c_m(\id:\ell_p^b(\ell_q^d) \to \ell_p^b(\ell_p^d))\leq b^{1/q-1/p}c_m(\id:\ell_q^b(\ell_q^d)\to \ell_p^b(\ell_p^d))\leq  c_{p,q} b^{1/q-1/p}\Big(\frac{\log(ebd/m)}{m}\Big)^{1/q-1/p}.$$

\noindent \textbf{Lower bound}: We prove the result by contradiction using Proposition~\ref{prop:GelEstPack}, so suppose that
$$c_m(\id:\ell_p^b(\ell_q^d) \to \ell_p^b(\ell_p^d))<c_{p,q}\mu(m)^{1/q-1/p},$$
where 
$$\mu(m):=\min\Big\{1,\frac{\tilde{c}_{p,q} b\log(ebd/m)}{64m}\Big\}$$
and $c_{p,q}$ and $\tilde{c}_{p,q}$ are constants depending only on $p$ and $q$ which are to be determined below. Set $t=\lfloor \mu(m)^{-1}\rfloor$. Note that $t\in \N$ and 
$$\frac{1}{2\mu(m)}<t\leq \frac{1}{\mu(m)}.$$
For $x\in \mathrm{ker}(A)$, $x\neq 0$, we have by assumption
$$\|x\|_{\ell_p(\ell_q)}\leq d^{1/q-1/p}\|x\|_{\ell_p(\ell_p)}\leq c_{p,q} d^{1/q-1/p} t^{-(1/q-1/p)}\|x\|_{\ell_p(\ell_q)}$$
so in particular,
\begin{equation}
\label{eqn:dovert}
d/t\geq c_{p,q}^{-1/(1/q-1/p)}.
\end{equation}
By choosing $c_{p,q}$ small enough, we can arrange that $t\leq d/8$, so that we can apply Proposition~\ref{prop:sparsePacking} below.\par
We apply Proposition~\ref{prop:GelEstPack} for $X=\ell_p^b(\ell_q^d)$ and $Y=\ell_p^b(\ell_p^d)$. Let 
$$\cS = \{[b]\ti\tilde{S} \ : \ \tilde{S}\subset [d], \ \card(\tilde{S})\leq 2t\}$$ 
be the set of all $2t$-inner sparse support sets, so that $\cS_+$ is the set of all $4t$-inner sparse support sets. By H\"{o}lder's inequality, we have the sharp inequality 
$$\|x_S\|_{\ell_p(\ell_q)}\leq (4t)^{1/q-1/p}\|x_S\|_{\ell_p(\ell_p)}, \qquad \text{for all } S\in \cS_+ \text{ and } x\in \R^{bd},$$
and therefore
$$\ka_{X,Y,\cS_+} = (4t)^{1/q-1/p}.$$
Hence, if
$$c_{p,q}<\Big(\frac{1}{4}\Big)^{1/q-1/p} (2\al_{\ell_p^b(\ell_q^d)})^{-1}$$
then by assumption
$$c_m(\id:\ell_p^b(\ell_q^d) \to \ell_p^b(\ell_p^d))<c_{p,q}\mu(m)^{1/q-1/p}\leq (2\al_{\ell_p^b(\ell_q^d)})^{-1} (4t)^{-(1/q-1/p)}.$$
By Proposition~\ref{prop:GelEstPack}, for any $W\subset \{x_S \ : \ S \in \cS\}$ and any $\eps>0$,
$$\log\cP(W, \|\cdot\|_{\ell_p^b(\ell_q^d)},\eps) \leq m\log\Big(\al_{\ell_p^b(\ell_q^d)}+\frac{2\al^2_{\ell_p^b(\ell_q^d)}\beta_{\ell_p^b(\ell_q^d)}r_{\ell_p^b(\ell_q^d)}(W)}{\eps}\Big).$$
Taking $W$ as in Proposition~\ref{prop:sparsePacking} with $s=b/32$ and setting $\eps=(b/32)^{1/p}(2t)^{1/q}$ yields
\begin{equation}
\label{eqn:dovertAppl}
\frac{bt}{32}\log\Big(\frac{d}{8t}\Big) \leq m\log(\al_{\ell_p^b(\ell_q^d)}+2^{1+1/p}\al^2_{\ell_p^b(\ell_q^d)}\beta_{\ell_p^b(\ell_q^d)}).
\end{equation}
Define
$$\tilde{c}_{p,q}^{-1} := \log(\al_{\ell_p^b(\ell_q^d)}+2^{1+1/p}\al^2_{\ell_p^b(\ell_q^d)}\beta_{\ell_p^b(\ell_q^d)}).$$
By (\ref{eqn:dovert}) we can pick $c_{p,q}$ small enough to arrange that 
$$\frac{1}{32}\log\Big(\frac{d}{8t}\Big)\tilde{c}_{p,q} \geq 8e$$
and, as a consequence, $m\geq 8ebt$. Combining this with (\ref{eqn:dovertAppl}) we find
\begin{align*}
m\geq \frac{\tilde{c}_{p,q}}{32} tb\log\Big(\frac{ebd}{m}\Big) & >\frac{\tilde{c}_{p,q}}{64}\mu(m)^{-1} b\log\Big(\frac{ebd}{m}\Big) \\
& = \frac{\tilde{c}_{p,q}}{64}b\log\Big(\frac{ebd}{m}\Big)\Big(\min\Big\{1,\frac{\tilde{c}_{p,q}b\log\Big(\frac{ebd}{m}\Big)}{64m}\Big\}\Big)^{-1} \\
& \geq m,
\end{align*}
which is the desired contradiction. 
\end{proof}
Finally, we prove statement (iii) in Theorem~\ref{thm:GelMain}. In contrast to the results in (i) and (ii), this is a truly mixed situation in which both the outer and the inner sparsity have a visible effect on the bounds.
\begin{proposition}
Set $0<q\leq 1$, $q\leq p\leq 1$ and $m\leq bd$. Then, there is a constant $\tilde{c}_{p,q}$, so that
\begin{align*}
 c_m(\id:\ell_p^b(\ell_q^d) \to \ell_2^b(\ell_2^d)) \simeq_{p,q}
 \begin{cases}
  1 &: 1 \leq m \leq \tilde{c}_{p,q}\log\Big(\frac{ebd}{m}\Big)\,,\\
  \left(\frac{\log(ebd/m)}{m}\right)^{1/p-1/2} &: \tilde{c}_{p,q}\log\Big(\frac{ebd}{m}\Big) \leq m \leq \tilde{c}_{p,q}b\log\Big(\frac{ebd}{m}\Big)\,,\\
  b^{1/2-1/p} \left(\frac{b\log(ebd/m)}{m}\right)^{1/q-1/2} &: \tilde{c}_{p,q}b\log\Big(\frac{ebd}{m}\Big) \leq m \leq bd\,.
 \end{cases}
\end{align*}
\end{proposition}
\begin{proof}
We define the constant $\tilde{c}_{p,q}$ explicitly in \eqref{eqn:cpqDefLast}. We distinguish three cases.\par
\emph{Case $1 \leq m \leq (\tilde{c}_{p,q}/2)\log(ebd/m)$}. The result trivially follows from the 
norm-$1$ embeddings $\mixedell{q}{q} \hookrightarrow \mixedell{p}{q} \hookrightarrow \mixedell{p}{p}$, (S3) and \eqref{eqn:Foucartetal}.\par
\emph{Case $(\tilde{c}_{p,q}/2)\log(ebd/m) \leq m \leq (\tilde{c}_{p,q}b/16)\log(ebd/m)$}. 
Again, the upper bound is immediate from $\|\id:\mixedell{p}{q}\to \mixedell{p}{p}\| = 1$, (S3) and \eqref{eqn:Foucartetal}. 
The lower bound follows using the proof technique used for Proposition~\ref{res:GelLowerBound}, 
but now setting $t=1$ instead of $t=d/(8e)$. We include the details for the reader's convenience.\par
We prove the lower bound by contradiction using Proposition~\ref{prop:GelEstPack}, so suppose that
$$c_m(\id:\ell_p^b(\ell_q^d)\to \ell_2^b(\ell_2^d))<c_{p,q}\mu(m)^{1/p-1/2},$$
where 
$$\mu(m):=\frac{\frac{1}{2}\tilde{c}_{p,q}\log(ebd/m)}{m}$$
and $c_{p,q},\tilde{c}_{p,q}$ are to be determined below. Set $s=\lfloor \mu(m)^{-1}\rfloor$. By our assumption
$$(\tilde{c}_{p,q}/2)\log(ebd/m) \leq m \leq (\tilde{c}_{p,q}b/16)\log(ebd/m),$$
it follows that $s\in \N$ and 
\begin{equation}
\label{eqn:boverspq}
s\leq \frac{1}{\mu(m)} = \frac{m}{\frac{1}{2}\tilde{c}_{p,q}\log(ebd/m)} \leq \frac{b}{8},
\end{equation}
so we can apply Proposition~\ref{prop:sparsePacking} further below. Moreover, for $x\in \mathrm{ker}(A)$, $x\neq 0$, we have by assumption
$$\|x\|_{\ell_p^b(\ell_q^d)}\leq b^{1/p-1/2}d^{1/q-1/2}\|x\|_{\ell_2^b(\ell_2^d)}\leq b^{1/p-1/2}d^{1/q-1/2}c_{p,q}s^{-(1/p-1/2)}\|x\|_{\ell_p^b(\ell_q^d)}$$
so in particular,
\begin{equation}
\label{eqn:boverspqandq}
b/s\geq c_{p,q}^{-1/(1/p-1/2)}d^{-\frac{1/q-1/2}{1/p-1/2}}.
\end{equation}
We apply Proposition~\ref{prop:GelEstPack} for $X=\ell_p^b(\ell_q^d)$ and $Y=\ell_2^b(\ell_2^d)$. Let $\cS$ be the set of all $(2s,2)$-sparse support sets, so that $\cS_+$ is the set of all $(4s,4)$-sparse support sets. By the sharp inequality 
$$\|x_S\|_{\ell_p(\ell_q)}\leq (4s)^{1/p-1/2} 4^{1/q-1/2} \|x_S\|_{\ell_2(\ell_2)}, \qquad \text{for all } S\in \cS_+ \text{ and } x\in \R^{bd},$$
it follows that
$$\ka_{X,Y,\cS_+} = (4s)^{1/p-1/2}4^{1/q-1/2}.$$
Therefore, if 
$$c_{p,q}\leq \Big(\frac{1}{4}\Big)^{1/p-1/2}\Big(\frac{1}{4}\Big)^{1/q-1/2}(2\al_{\ell_p^b(\ell_q^d)})^{-1},$$
then by assumption
$$c_m(\id:\ell_p^b(\ell_q^d)\to \ell_2^b(\ell_2^d))<c_{p,q}\mu(m)^{1/p-1/2}\leq (2\al_{\ell_p^b(\ell_q^d)})^{-1}(4s)^{-(1/p-1/2)}4^{-(1/q-1/2)}.$$
By Proposition~\ref{prop:GelEstPack}, for any $W\subset \{x_S \ : \ S \in \cS\}$ and any $\eps>0$,
$$\log\cP(W, \|\cdot\|_{\ell_p^b(\ell_q^d)},\eps) \leq m\log\Big(\al_{\ell_p^b(\ell_q^d)}+\frac{2\al_{\ell_p^b(\ell_q^d)}^2\beta_{\ell_p^b(\ell_q^d)}r_{\ell_p^b(\ell_q^d)}(W)}{\eps}\Big).$$
By Proposition~\ref{prop:sparsePacking}, applied with $t=1$, there is a set $W$ of $(2s,2)$-sparse vectors with $r_{\ell_p^b(\ell_q^d)}(W)\leq (2s)^{1/p}2^{1/q}$ and 
$$\cP\Big(W, \|\cdot\|_{\ell_p^b(\ell_q^d)},s^{1/p}2^{1/q}\Big)\geq \Big(\frac{b}{32s}\Big)^s \Big(\frac{d}{8}\Big)^{s}=\Big(\frac{bd}{256s}\Big)^s.$$
This implies
\begin{equation}
\label{eqn:boversApplpq}
s\log(bd/256s)\leq m\log(\al_{\ell_p^b(\ell_q^d)}+2^{1+1/p}\al_{\ell_p^b(\ell_q^d)}^2\beta_{\ell_p^b(\ell_q^d)}).
\end{equation}
Setting
\begin{equation}
\label{eqn:cpqDefLast}
\tilde{c}_{p,q}:=\Big(\log(\al_{\ell_p^b(\ell_q^d)}+2^{1+1/p}\al_{\ell_p^b(\ell_q^d)}^2\beta_{\ell_p^b(\ell_q^d)})\Big)^{-1},
\end{equation}
it follows from (\ref{eqn:boverspq}) that $(b/s)\geq 8$ and hence
$$m\geq \tilde{c}_{p,q} s\log(bd/256s) \geq \tilde{c}_{p,q} s\log(d/32).$$
Therefore, if 
$$d\geq 32 e^{\tilde{c}_{p,q}^{-1}256e}=:d_{p,q},$$
then $m\geq 256es$. On the other hand, if $d\leq d_{p,q}$, then by \eqref{eqn:boverspqandq}
$$b/s\geq c_{p,q}^{-1/(1/p-1/2)}d^{-\frac{1/q-1/2}{1/p-1/2}} \geq c_{p,q}^{-1/(1/p-1/2)}d_{p,q}^{-\frac{1/q-1/2}{1/p-1/2}}.$$
Hence, by taking $c_{p,q}$ small enough, we can ensure that
$$m\geq \tilde{c}_{p,q} s\log(bd/256s)\geq \tilde{c}_{p,q} s\log(b/256s) \geq 256es.$$
Using this in (\ref{eqn:boversApplpq}) yields
$$m\geq \tilde{c}_{p,q} s\log(ebd/m)$$
Since $s>\frac{1}{2\mu(m)}$, we now find
$$m\geq \tilde{c}_{p,q} s\log(ebd/m) >\frac{1}{2}\mu(m)^{-1} \tilde{c}_{p,q} \log(ebd/m) = m,$$
which is the desired contradiction.\par
\emph{Case $(\tilde{c}_{p,q}b/16)\log(ebd/m) \leq m \leq bd$}. Both, the lower and the upper bound follow from factorization of the identity and (S3). Namely, 
$\|\id:\ell_p^b(\ell_q^d) \to \ell_q^b(\ell_q^d)\| = b^{1/q-1/p}$ and $\|\id:\ell_2^b(\ell_q^d) \to \ell_p^b(\ell_q^d)\| = b^{1/p-1/2}$. Part (i) and (ii) of Theorem~\ref{thm:GelMain} 
together with (S3) imply that
\begin{align*}
b^{1/2-1/p}\min\Big\{1,\frac{b\log(ebd/m)}{m}\Big\}^{1/q-1/2} & \simeq_{p,q} b^{1/2-1/p} c_m(\id:\ell_2^b(\ell_q^d) \to \ell_2^b(\ell_2^d)) \\
& \leq c_m(\id:\ell_p^b(\ell_q^d) \to \ell_2^b(\ell_2^d)) \\
& \leq b^{1/q-1/p} c_m(\id:\ell_q^b(\ell_q^d) \to \ell_2^b(\ell_2^d)) \\
& \simeq_{p,q} b^{1/q-1/p} \min\Big\{1,\frac{\log(ebd/m)}{m}\Big\}^{1/q-1/2}.
\end{align*}
The lower and upper bound match if $m\geq b\log(ebd/m)$. 
\end{proof} 
\begin{remark}
There are several remaining open cases for which it would be interesting to prove matching bounds for the Gelfand numbers: 

{\em (i)} It would be very interesting to extend our result in \eqref{eqn:GelMain1} 
to the case of differing inner spaces. That is, one would like to have optimal bounds for 
$c_m(\id:\ell_p^b(\ell_q^d)\to \ell_r^b(\ell_u^d))$ for $p\leq 1$, $p<r$ and $q<u$. 
As we discuss in Remark~\ref{rem:diffInner} in more detail, the upper bounds are highly relevant for our application to 
Besov space embeddings in Section~\ref{sec:appl}. Establishing the lower bounds in the case $p\leq q\leq 1$ is 
interesting from a purely mathematical perspective, as it requires to understand the precise relation with structured sparsity in this parameter range.    

{\em(ii)} It would also be interesting to extend the result in (iii) of Theorem~\ref{thm:GelMain} to the case $0<q\leq 1$, $1<p\leq 2$. 
As has been mentioned, the two-sided bounds in the first case (for small $m$) as well 
as in the third case (for large $m$) remain valid, but the bounds in the second case cannot be correct 
(this can be seen by taking $d=1$). From the perspective of classical $\ell_p^n$-spaces, 
this is a situation where one expects a mixture of the behaviors of $c_m(\id:\ell_p^n\to \ell_2^n)$ for $p\leq 1$ and $p>1$.

{\em (iii)} In addition, we expect that it is possible to extend our results to mixed $\ell_p(\ell_q)$-spaces of block vectors $(x_i)_{i=1}^b$ whose blocks have varying sizes. This would be of interest to extend our fundamental lower bounds for structured sparse recovery in Corollary~\ref{cor:stabRecIntro}. Such an extension would be particularly interesting for recovery of sparse-in-levels vectors, since in a typical application the block sizes are varying (see e.g.\ \cite{AHP14,BaH14,BBW17,LiA16}).   
\end{remark}

\section{Applications}
\label{sec:appl}

\subsection{Signal processing}

From the perspective of signal processing, Gelfand widths (see Remark \ref{rem:Gelfand_width}) are of interest due to their close connection to the compressive $m$-widths, defined in \eqref{eqn:compressWidth} above. As was discussed in Section~\ref{sec:GelEquiv}, we have 
$$c_m(B_{\ell_p^b(\ell_q^d)},\ell_r^b(\ell_u^d)) \simeq_{p,q,r,u} E_m(B_{\ell_p^b(\ell_q^d)},\ell_r^b(\ell_u^d)).$$ 
Thanks to this equivalence, we can use Theorem~\ref{thm:GelSimple} to derive Corollary~\ref{cor:stabRecIntro}. Its consequences for structured signal recovery were discussed before in the introduction.
\begin{proof}[Proof of Corollary~\ref{cor:stabRecIntro}]
By (\ref{eqn:stabBlock}), 
$$\sup_{x\in B_{\ell_1^b(\ell_2^d)}} \|x-\Del(Ax)\|_{\ell_2^b(\ell_2^d)} \leq D\sup_{x\in B_{\ell_1^b(\ell_2^d)}} \frac{\si_s^{\operatorname{outer}}(x)_{\ell_1^b(\ell_2^d)}}{\sqrt{s}} \leq \frac{D}{\sqrt{s}}.$$
Theorem~\ref{thm:GelSimple} therefore implies that for some $c>0$ and $C\geq 1$,
$$c\min\Big\{1,\frac{Ce(\log(eb/m)+d)}{m}\Big\}^{1/2} \leq c_m(\id: \ell_1^b(\ell_2^d) \to \ell_2^b(\ell_2^d)) \leq \frac{D}{\sqrt{s}}.$$
This means that either
$$s\leq \frac{D^2}{c^2} \qquad \text{or} \qquad s\leq \frac{D^2}{c^2}\frac{m}{Ce(\log(eb/m)+d)}.$$
Thus, if we assume that $s\geq D^2/c^2$, then we obtain
$$s\leq \frac{D^2}{c^2}\frac{m}{Ce(\log(eb/m)+d)} = \frac{D^2}{c^2Ce}\frac{m}{\log(ebe^d/m)} $$
We now apply Lemma~\ref{lem:invert} with $y=m$, $x=c^2 s/D^2$ and $K=be^d$ to obtain
$$m\geq \frac{Ce}{1+\log(Ce)} \frac{c^2 s}{D^2}\log(eD^2be^d/c^2s)\gtrsim s\log(eb/s) + sd.$$
Suppose now that (\ref{eqn:stabInner}) holds. Then,
$$\sup_{x\in B_{\ell_2^b(\ell_1^d)}} \|x-\Del(Ax)\|_{\ell_2^b(\ell_2^d)} \leq D\sup_{x\in B_{\ell_2^b(\ell_1^d)}}\frac{\si_t^{\operatorname{inner}}(x)_{\ell_2^b(\ell_1^d)}}{\sqrt{t}}\leq \frac{D}{\sqrt{t}}.$$
Theorem~\ref{thm:GelSimple} now yields for some $c>0$ and $C\geq 1$,
$$c \min\Big\{1,\frac{Ceb\log(ebd/m)}{m}\Big\}^{1/2} \leq c_m(\id:\ell_2^b(\ell_1^d) \to \ell_2^b(\ell_2^d))\leq \frac{D}{\sqrt{t}}.$$
Therefore, either
$$t\leq \frac{D^2}{c^2} \qquad \text{or} \qquad t\leq \frac{D^2}{c^2}\frac{m}{Ceb\log(ebd/m)}.$$
Assuming that $t\geq D^2/c^2$, we find
$$t\leq \frac{D^2}{c^2Ce}\frac{(m/b)}{\log(ed/(m/b))}$$ 
We now apply Lemma~\ref{lem:invert} with $y=m/b$, $x=c^2 t/D^2$ and $K=d$ to obtain
$$\frac{m}{b} \geq  \frac{Ce}{1+\log(Ce)} \frac{c^2 t}{D^2}\log(eD^2bd/c^2t)$$
and as a consequence, 
$$m\gtrsim bt\log(ed/t).$$
\end{proof}
In the introduction we claimed that if $B$ is standard Gaussian and $A=(1/\sqrt{m})B$, then the pair $(A,\Del_{\ell_1(\ell_1)})$ with high probability satisfies the stable reconstruction guarantee \eqref{eqn:stabInner}, provided that $m\gtrsim bt\log(ed/t)$. By Corollary~\ref{cor:stabRecIntro} this is optimal. Here we provide details on this claim for the convenience of the reader. We know that $(A,\Del_{\ell_1(\ell_1)})$ with high probability satisfies (\cite[Theorem 1.2 and Section 1.3]{CRT06}, see also \cite{CaT06,do06_2} and \cite[Theorem 9.13]{FoR13})
\begin{equation}
\label{eqn:stab11Appl}
\|x-\Del_{\ell_1(\ell_1)}(Ax)\|_{\ell_2^b(\ell_2^d)} \leq D\frac{\si_k(x)_{\ell_1^b(\ell_1^d)}}{\sqrt{k}},
\end{equation}
whenever $m\gtrsim k\log(ebd/k)$, where 
$$\si_k(x)_{\ell_1^b(\ell_1^d)} = \inf\{\|x-z\|_{\ell_1^b(\ell_1^d)} \ : \ z \ \operatorname{is} \ k\operatorname{-sparse}\}.$$
Set $k=bt$. By H\"older's inequality and by noting that every $t$-inner sparse vector is $bt$-sparse, we find
\begin{align*}
\si_{bt}(x)_{\ell_1^b(\ell_1^d)} & = \inf\{\|x-z\|_{\ell_1^b(\ell_1^d)} \ : \ z \ \operatorname{is} \ bt\operatorname{-sparse}\} \\
& \leq \sqrt{b} \inf\{\|x-z\|_{\ell_2^b(\ell_1^d)} \ : \ z \ \operatorname{is} \ bt\operatorname{-sparse}\} \\
& \leq \sqrt{b} \inf\{\|x-z\|_{\ell_2^b(\ell_1^d)} \ : \ z \ \operatorname{is} \ t\operatorname{-inner \ sparse}\} = \sqrt{b} \si_t^{\operatorname{inner}}(x)_{\ell_2^b(\ell_1^d)}.
\end{align*}
Therefore, by \eqref{eqn:stab11Appl}, 
\begin{equation*}
\|x-\Del_{\ell_1(\ell_1)}(Ax)\|_{\ell_2^b(\ell_2^d)} \leq D\frac{\si_{bt}(x)_{\ell_1^b(\ell_1^d)}}{\sqrt{bt}} \leq D \frac{\si_t^{\operatorname{inner}}(x)_{\ell_2^b(\ell_1^d)}}{\sqrt{t}},
\end{equation*}
provided that $m\gtrsim k\log(ebd/k) = bt\log(ed/t)$.

\subsection{Gelfand numbers of Besov space embeddings with small mixed smoothness}
\label{Sect7.2}
The above results in finite dimensional mixed-norm spaces have a direct application to Besov space embeddings with small 
mixed smoothness. In this section we are interested in the Gelfand numbers of the embeddings 
\begin{equation}\label{emb2}
    \text{Id}:S^{r_0}_{p_0,q_0}B(\Omega) \to S^{r_1}_{p_1,q_1}B(\Omega)\,,
\end{equation}
where $\Omega\subset \R^d$ is a bounded domain (open set) and $0<p_0 \leq p_1 \leq \infty$ and $0<q_0 < q_1 \leq \infty$ such that we are in the small smoothness
regime, i.e. $1/p_0-1/p_1 < r_0-r_1 \leq 1/q_0-1/q_1$\,. The goal of this section is to show that the results in Theorem \ref{thm:GelMain} imply 
\begin{equation}\label{loglog}
    c_m(\text{Id}:S^{r_0}_{p_0,q_0}B(\Omega) \to S^{r_1}_{p_1,q_1}B(\Omega)) \simeq_{\theta} m^{-(r_0-r_1)}
\end{equation}
in some of the parameter constellations in the small smoothness regime. Here and below we write 
\begin{equation}
\label{def:theta}
\theta:=\{p_0,p_1,q_0,q_1,r_0,r_1,d,\Omega\}
\end{equation}
to abbreviate the dependency on the various parameters.\par  
Observe that the decay behavior in \eqref{loglog} is typical for the univariate situation ($d=1$). Usually, in the $d$-variate setting we encounter asymptotic orders such as 
$m^{-(r_0-r_1)}(\log m)^{(d-1)\eta}$
for some $\eta:=\eta(p_0,q_0,p_1,q_1,r_0,r_1)$, see Remark \ref{large_s} and \cite{DTU16,Kien16}. The dimension $d$ of the underlying Euclidean space enters the rate of convergence exponentially. Surprisingly, this is no longer the case in the small smoothness regime as Corollaries \ref{L2}, \ref{endpoint} and \ref{nonsharp} below will show. However, we emphasize that dimension $d$ is still hidden in the constants. 

We start with some definitions. We denote by $L_p(\Omega)$, $0<p\leq \infty$, the space of all measurable complex-valued functions $f:\Omega\to \C$ where
$$\|f\|_p:=\Big(\int_{\Omega} |f(x)|^p dx\Big)^{1/p}$$ is finite (with the usual modification if $p=\infty$). The space of tempered distributions on $\R^d$ is 
denoted with $S'(\R^d)$. For $0<p,q\leq \infty$ and $r\in \R$ the Besov space $S^r_{p,q}B(\R^d)$ is defined as
\begin{equation}\label{defB}
    S^r_{p,q}B(\R^d) := \Big\{f\in S'(\R^d)~:~\|f\|_{S^r_{p,q}B}:=\Big(\sum\limits_{\bar{j} \in \N_0^d}
    2^{rq\|\bar{j}\|_1}\|\cF^{-1}[\varphi_{\bar{j}}\cF f]\|_p^q\Big)^{1/q} < \infty\Big\}\,,
\end{equation}
(with the usual modification in case $q=\infty$) where the system $\{\varphi_{\bar{j}}\}_{\bar{j}\in \N_0^d}$ is the standard tensorized 
dyadic decomposition of unity, see \cite{Vyb06} and \cite{ScTr87}. Important for our purpose is the following
diagonal embedding result, see \cite[Chapt.\ 2]{ScTr87},
\begin{equation}\label{emb}
   S^{r_0}_{p_0,q_0}B(\R^d) \hookrightarrow S^{r_1}_{p_1,q_1}B(\R^d)
\end{equation}
for $p_0 \leq p_1$ and $r_0-r_1>1/p_0-1/p_1$\,. The embedding \eqref{emb} is never compact which is the reason why 
we will restrict to spaces on bounded domains in the sequel. This guarantees a decay of the corresponding Gelfand numbers of the embedding. 
Before doing so, let us comment on the discretization of the above function
spaces on $\R^d$ developed in \cite[Thm.\ 2.10]{Vyb06}. Using sufficiently smooth wavelets 
with sufficiently many vanishing moments (and the notation from \cite{Vyb06}) the mapping
\begin{equation}\label{wave_iso}
    f \mapsto \lambda_{\bar{j},\bar{k}}(f):=\langle f, \psi_{\bar{j},\bar{k}} \rangle,\qquad \bar{j}\in \N_0^d, 
    \bar{k} \in \Z^d\,,
\end{equation}
represents a sequence spaces isomorphism between $S^r_{p,q}B(\R^d)$ and
\begin{equation}\label{sequ_space}
    s^r_{p,q}b := \Big\{\lambda=\{\lambda_{\bar{j},\bar{k}}\}_{\bar{j}, \bar{k}} \subset \C~:~
    \|\lambda\|_{s^r_{p,q}b}:=\Big[\sum\limits_{\bar{j}\in \N_0^d} 2^{(r-1/p)q\|\bar{j}\|_1}
    \Big(\sum\limits_{\bar{k}\in \Z^d}|\lambda_{\bar{j},\bar{k}}|^p\Big)^{q/p}\Big]^{1/q}<\infty\Big\}
\end{equation}
(with the usual modification in case $\max\{p,q\} = \infty$)\,. We want to study Gelfand numbers of compact 
function space embeddings on domains. Let $\Omega$ be an arbitrary bounded domain in $\R^d$, $D(\Omega)$ denote 
the set of test functions and $D'(\Omega)$ the corresponding set of distributions. Then $S^r_{p,q}B(\Omega)$ is defined as the set of restrictions 
(in the distributional sense) of elements of $S^r_{p,q}B(\R^d)$ to $\Omega$, i.e.,
$$
    S^r_{p,q}B(\Omega) := \{f\in D'(\Omega)~:~\exists g\in S^r_{p,q}B(\R^d) \text{ such that } g|_{\Omega} = f\}\,
$$
and its (quasi-)norm is given by $\|f\|_{S^r_{p,q}B(\Omega)}:=\inf_{g|_{\Omega} = f} \|g\|_{S^r_{p,q}B}$, see 
\cite[Sect.\ 3.1]{Vyb06}, \cite[3.2.3]{HaSi12}\,. 
The embedding \eqref{emb} transfers to the domain $\Omega$ and is compact so that the Gelfand numbers 
decay and converge to zero. We are interested in establishing the decay rate especially for the embedding \eqref{emb2} in the important special case $\Omega = [0,1]^d$, which is representative for all tensor product domains.  

Using the boundedness of certain restriction and extension operators, see \cite[Sect.\ 4.5]{Vyb06}, \cite[3.2.3]{HaSi12}, 
together with the wavelet isomorphism \eqref{wave_iso} between \eqref{defB} and \eqref{sequ_space}, we can restrict ourselves to the study of embeddings between the ``restricted sequence spaces''
$s^r_{p,q}b(\Omega)$. For $\bar j \in \N_0^d$ and $\bar{k} \in \Z^d$ let $Q_{\bar j, \bar k} = \times_{i=1}^d 2^{-j_i}[k_i-1,k_i+1]$. 
Let $A_{\bar j}^\Omega = \{\bar k \in \Z^d: Q_{\bar j, \bar k} \cap \Omega \neq \emptyset \}$. Consider the sequence space
$$
s_{p,q}^{r}b(\Omega) := \{ \lambda = (\lambda_{\bar j, \bar k})_{\bar j \in \N_0^d, \bar k \in A_{\bar j}^\Omega}\subset \C : \|\lambda\|_{s_{p,q}^{r}b(\Omega)} < \infty  \}
$$
with (quasi-)norm given by
\begin{equation}\label{sequBOmega} 
    \|\lambda\|_{s_{p,q}^{r}b(\Omega)} := \Big[\sum_{\bar j \in \N_0^d} 2^{\|\bar j\|_1(r-1/p)q}
    \Big( \sum_{\bar k \in A_{\bar j}^\Omega} |\lambda_{\bar j, \bar k}|^p \Big)^{q/p} \Big]^{1/q}.
\end{equation}
 
Let us start with the following result, where the upper and lower bound only differ by a power of $\log\log m$.

\begin{proposition}\label{res:appl_embedding_entropy0} Let $0<q_0\leq p_0 \leq 1$, $0<p_0 \leq p_1\leq  q_1 \leq 2$ and $1/p_0-1/p_1<r \leq 1/q_0-1/q_1$. Let $\theta$ be as defined in \eqref{def:theta}. For $m\geq 2$ we have 
(with $\varrho:=\min\{1,p_1,q_1\}$)
\begin{equation}\label{res1}
  \begin{split}
      m^{-r} &\lesssim_\theta c_m(\id:s^r_{p_0,q_0}b(\Omega) \to s^0_{p_1,q_1}b(\Omega)) \\
      &\lesssim_\theta m^{-r}(\log\log_2 m)^{1/q_0-1/q_1}\left\{\begin{array}{rcl}
	(\log\log_2 m)^{r+1/\varrho}&:&r=1/q_0-1/q_1\,,\\
	1&:&r<1/q_0-1/q_1\,.
      \end{array}\right.
  \end{split}    
\end{equation}
 
\end{proposition}

\begin{proof} {\em Step 1.} By the monotonicity and subadditivity of Gelfand numbers (see (S1) and (S2) above) we may decompose as follows
\begin{equation} \label{decomp}
 c_m(\id)^{\varrho} \leq \sum_{\mu=0}^\infty c_{m_\mu+1}(\id_\mu)^{\varrho},
\end{equation}
where $\id = \sum_{\mu=0}^{\infty} \id_{\mu}$ with 
\[ 
 (\id_\mu \lambda)_{\bar j, \bar k} = \begin{cases}
 \lambda_{\bar j, \bar k} &: \|\bar j\|_1 = \mu, \bar{k}\in A^{\Omega}_{\bar{j}}\,,\\
 0 &: \text{ otherwise, }
 \end{cases}
\]
and $(m_\mu)_{\mu=0}^\infty$ is a sequence of natural numbers $m_\mu \in \N_0$ such that $m = \sum_{\mu=1}^{\infty} m_\mu < \infty$. Similar to \cite[(3.8)]{Vyb06}
we define the quantity 
$$
    D_\mu:=\sum\limits_{\|\bar j\|_1 = \mu} \# (A_{\bar{j}}^{\Omega})\,,
$$
which denotes the dimension of the block with number $\mu$\,. Clearly, 
\begin{equation}\label{count}
  \sharp \{\bar j \in \N_0^d: \|\bar j\|_1 = \mu \} = {\mu + d -1 \choose \mu} \simeq_d \mu^{d-1}\quad \mbox{and} \quad \sharp A_{\bar j}^\Omega \simeq_{\Omega} 2^{\|\bar j\|_1}
\end{equation}
such that $D_\mu \simeq_{d,\Omega} 2^\mu \mu^{d-1}$. 
Now we fix some $J \in \N$ and consider four different ranges in the sum on 
the right-hand side in \eqref{decomp}, i.e., 
\begin{equation} \label{decomp2}
 c_m(\id)^{\varrho} \leq \sum_{\mu=0}^J c_{m_\mu+1}(\id_\mu)^{\varrho} + \sum_{\mu=J+1}^L c_{m_\mu+1}(\id_\mu)^{\varrho} + \sum_{\mu=L+1}^M c_{m_\mu+1}(\id_\mu)^{\varrho}
 + \sum_{\mu=M+1}^\infty \|\id_\mu\|^{\varrho}\,.
\end{equation}
Here we put $L = J+(d-1)\log_2 J$ and $M$ will depend on $L$ as determined below. The main challenge is to choose the $m_\mu$ suitably and estimate the 
corresponding $c_{m_\mu}(\id_\mu)$ using the finite dimensional results from Section \ref{sec:matchGel}. Note, that we deal with $\ell_q(\ell_p)$ spaces of complex 
numbers here. This technical issue only causes a multiplicative constant in the estimates below.

{\em Step 2.} To estimate the first sum we choose $m_\mu = 2D_\mu$ which yields (see (S4))
\begin{equation}\label{est0}
   \sum\limits_{\mu=0}^J c_{m_\mu+1}(\id_\mu)^{\varrho} = 0 \,.
\end{equation}
In addition, $\sum_{\mu=0}^J m_{\mu} \simeq_{d,\Omega} 2^JJ^{d-1}$\,.
In case $m_{\mu} \leq 2^{\mu}\mu^{d-1}$ we have as a consequence of \eqref{count}
\begin{align}\label{eq:fundamental_block_equivalence} 
 c_{m_\mu}(\id_\mu) \lesssim_{\theta} 2^{(-r+1/p_0-1/p_1)\mu} \; 
 c_{m_\mu} \left(\id: \ell_{q_0}^{\mu^{d-1}}(\ell_{p_0}^{2^\mu}) \to \ell_{q_1}^{\mu^{d-1}}(\ell_{p_1}^{2^\mu}) \right)\,.
\end{align}
If $m_{\mu} > 2^{\mu}\mu^{d-1}$ then \eqref{eq:fundamental_block_equivalence} may fail since the right-hand side vanishes due to (S4). 
Using straightforward embedding arguments for finite-dimensional mixed-norm spaces, the Gelfand numbers on the right-hand side can be estimated in two different ways using \eqref{eqn:Foucartetal}. 
On the one hand, we get 
\begin{equation}\label{est1}
 \begin{split}
    c_{m_\mu}(\id_\mu) &\lesssim_{\theta} 2^{-(r-1/p_0+1/p_1)\mu}2^{\mu(1/q_0-1/p_0)}2^{\mu(1/p_1-1/q_1)}c_{m_{\mu}}(\id:\ell^{2^{\mu}\mu^{d-1}}_{q
    _0} \to \ell^{2^{\mu}\mu^{d-1}}_{q_1})\\
    &\lesssim_{\theta} 2^{-r\mu}2^{\mu(1/q_0-1/q_1)}\Big(\frac{\log(e2^\mu \mu^{d-1}/m_{\mu})}{m_{\mu}}\Big)^{1/q_0-1/q_1}\,.
 \end{split}
\end{equation}
On the other hand, 
\begin{equation}\label{est2}
 \begin{split}
    c_{m_\mu}(\id_\mu) &\lesssim_{\theta} 2^{-(r-1/p_0+1/p_1)\mu}c_{m_{\mu}}(\id:\ell^{2^{\mu}\mu^{d-1}}_{p
    _0} \to \ell^{2^{\mu}\mu^{d-1}}_{p_1})\\
    &\lesssim_{\theta} 2^{-(r-1/p_0+1/p_1)\mu}\Big(\frac{\log(e2^\mu \mu^{d-1}/m_{\mu})}{m_{\mu}}\Big)^{1/p_0-1/p_1}\,.
 \end{split}
\end{equation}
Let us assume for the moment 
that $r<1/q_0-1/q_1$. In the second sum in \eqref{decomp2} we choose $m_\mu = 2^\mu 2^{(L-\mu)\kappa}$ with $\kappa<1$ such that $r<\kappa(1/q_0-1/q_1)$.
Clearly, $m_{\mu} \leq 2^{\mu}\mu^{d-1}$ and, due to $\kappa<1$, $\sum_{\mu=J+1}^L m_\mu \simeq_{\kappa} 2^L = 2^JJ^{d-1}$. Employing \eqref{est1} gives 

\begin{equation}\label{est3}
    \sum\limits_{\mu=J+1}^{L} c_{m_\mu+1}(\id_\mu)^{\varrho} 
    \lesssim_{\theta} \sum\limits_{\mu=J+1}^{L} 2^{-\mu r \varrho}\Big[2^{(\mu-L)\kappa}(\log e\mu^{d-1}2^{(\mu-L)\kappa})\Big]^{(1/q_0-1/q_1)\varrho}\,.
\end{equation}
Since we assumed $r<1/q_0-1/q_1$, we obtain 
\begin{equation}\label{est6}
    \sum\limits_{\mu=J+1}^{L} c_{m_\mu+1}(\id_\mu)^{\varrho} \lesssim_{\theta} 2^{-Lr\varrho}(\log L)^{(1/q_0-1/q_1)\varrho} \simeq_{\theta} (2^J J^{d-1})^{-r\varrho}(\log J)^{(1/q_0-1/q_1)\varrho}\,.
\end{equation}
In the case $r=1/q_0-1/q_1$ we choose $m_\mu = 2^J\mu^{d-1}$ in the range $J<\mu \leq L$. This gives $m_\mu \leq 2^\mu \mu^{d-1}$ and 
\begin{equation}\label{est31}
 \sum_{\mu=J+1}^L m_\mu \simeq_{d} 2^JJ^{d-1} \sum_{\mu=J+1}^L 1 \simeq_{d} 2^JJ^{d-1}\log J\,. 
\end{equation}
Instead of \eqref{est3} we obtain
\begin{equation}\label{est6b}
  \begin{split}
    \sum\limits_{\mu=J+1}^{L} c_{m_\mu+1}(\id_\mu)^{\varrho} &\lesssim_{\theta} 
    \sum\limits_{\mu=J+1}^L \Big[2^{-J}\mu^{-(d-1)}(\mu-J)\Big]^{(1/q_0-1/q_1)\varrho}\\
    &\simeq_{\theta} \Big[2^{J}J^{d-1}\Big]^{-r\varrho}(\log J)^{(1/q_0-1/q_1)\varrho}\log J\,.
  \end{split}  
\end{equation}

Now we discuss the third sum in \eqref{decomp2}, which we only need in case $p_0<p_1$ otherwise we continue directly 
with \eqref{est111} below where $M = L$. Here we use $m_\mu := \lfloor 2^\mu 2^{(L-\mu)\beta} \rfloor$ where $\beta>1$ is chosen
such that $r>\beta(1/p_0-1/p_1)$\,. Again $m_\mu \leq 2^\mu \mu^{d-1}$ and $\sum_{\mu=L}^{\infty} m_\mu \simeq_{\beta} 2^L = 2^JJ^{d-1}\,$. In addition, 
for $\mu>M:=\lceil L\beta/(\beta-1) \rceil$ we have $m_\mu=0$. Using
\eqref{est2} we obtain
\begin{equation}\label{est4}
 \begin{split}
    \sum\limits_{\mu = L+1}^{M} c_{m_\mu+1}(\id_\mu)^{\varrho} &\lesssim_{\theta} 
    \sum\limits_{\mu=L+1}^{M}2^{-(r-1/p_0+1/p_1)\mu\varrho}
    \Big(\frac{\log(e2^\mu \mu^{d-1}/m_{\mu})}{m_{\mu}}\Big)^{(1/p_0-1/p_1)\varrho}\\
    &\lesssim_{\theta} \sum\limits_{\mu = L+1}^{\infty} 2^{-r\mu\varrho}
    \Big[2^{(\mu-L)\beta}\log(2^{(\mu-L)\beta}\mu^{d-1})\Big]^{(1/p_0-1/p_1)\varrho}\\
    &\lesssim_{\theta} \sum\limits_{\mu = L+1}^{\infty} 
    2^{-r\mu\varrho}
    \Big[2^{(\mu-L)\beta}((\mu-L)\beta + \log \mu)\Big]^{(1/p_0-1/p_1)\varrho}\\
    &\simeq_{\theta} 2^{-rL\varrho} (\log L)^{(1/p_0-1/p_1)\varrho}\\
    &\simeq_{\theta} (2^JJ^{d-1})^{-r\varrho}(\log J)^{(1/p_0-1/p_1)\varrho}\,.
 \end{split}
\end{equation}

Let us finally care for the fourth sum in \eqref{decomp2}. It holds
\begin{equation}\label{est41}
 \begin{split}
    \sum\limits_{\mu = M+1}^{\infty} \|\id_\mu\|^{\varrho} &\lesssim_{\theta} 
    \sum\limits_{\mu=M+1}^{\infty}2^{-(r-1/p_0+1/p_1)\mu\varrho}\\
    &\simeq_{\theta} 2^{-(r-1/p_0+1/p_1)M\varrho}\\
    &\simeq_{\theta} 2^{-L(r-1/p_0+1/p_1)\beta/(\beta-1)\varrho}\,.
 \end{split}
\end{equation}
Note, that $r>\beta(1/p_0-1/p_1)$ for $\beta >1$ implies $(r-1/p_0+1/p_1)\beta/(\beta-1)>r$. Then it holds 
\begin{equation}\label{est42}
 \sum\limits_{\mu = M+1}^{\infty} \|\id_\mu\|^{\varrho} \lesssim_{\theta} 2^{-Lr\varrho} \simeq_{\theta} (2^JJ^{d-1})^{-r\varrho}\,.
\end{equation}

Altogether, \eqref{est0}, \eqref{est6}, \eqref{est4}, and \eqref{est42} yield in case $1/p_0-1/p_1 < r < 1/q_0-1/q_1$
\begin{equation}\label{est7}
  \begin{split}
    c_m(\id:s^r_{p_0,q_0}b(\Omega) \to s^0_{p_1,q_1}b(\Omega)) &\lesssim_{\theta} (2^J J^{d-1})^{-r}(\log J)^{1/q_0-1/q_1}\\
    &\simeq_{\theta} m^{-r} (\log\log_2 m)^{1/q_0-1/q_1}\,,
  \end{split}
\end{equation}
by taking into account 
\begin{equation}\label{est8}
    m_J := \sum\limits_{\mu=0}^{\infty} m_{\mu} \simeq_{\theta} 2^J J^{d-1}
\end{equation}
and a standard monotonicity argument to fill the gaps between $\{m_J\}_{J=1}^{\infty}$. Similarly, in the endpoint case $r=1/q_0-1/q_1$ the estimates \eqref{est31}  and \eqref{est6b} imply 
\begin{equation}\label{est71}
    c_m(\id:s^r_{p_0,q_0}b(\Omega) \to s^0_{p_1,q_1}b(\Omega)) \lesssim_{\theta} m^{-r} (\log\log_2 m)^{1/q_0-1/q_1}(\log\log_2 m)^{r+1/\varrho}\,.
\end{equation}

{\em Step 3 (lower bound).} The argument 
given by Vyb\'iral in \cite[Thm.\ 3.18, Step.\ 2]{Vyb06} for proving the same lower bound for entropy numbers almost literally 
applies to Gelfand numbers by taking the lower bound in Theorem \ref{thm:GelMain}, \eqref{eqn:Foucartetal} into account
(instead of \cite[Lem.\ 3.11]{Vyb06}). \end{proof}

In case $p_0 = p_1 = 2$ we obtain the following sharp result which turns out to be a consequence of our new 
improved block result in Theorem \ref{thm:GelMain}, \eqref{eqn:GelMain1}. The main benefit is that we do not produce the term $(\log\log_2 m)^{1/q_0-1/q_1}$, see 
\eqref{res1}, which is an artefact of the log-term in \eqref{eqn:Foucartetal}\,. Compared to Proposition~\ref{res:appl_embedding_entropy0}, we can get beyond the restriction $p_0\leq 1$, but now need to impose $p_0=p_1=2$\,.

\begin{proposition}\label{res:appl_embedding_entropy}
Let $0<q_0\leq 1$, $q_0<q_1 \leq 2$ and $0<r < 1/q_0-1/q_1$. Let further $\varrho = \min \{1,q_1\}$. For $m\geq 2$, we have
\begin{equation}\label{res2}
      c_m(\id:s^r_{2,q_0}b(\Omega) \to s^0_{2,q_1}b(\Omega)) \simeq_{\theta} m^{-r}\,.
\end{equation}
\end{proposition}

\begin{proof} We return to \eqref{decomp2} where we drop the third sum, i.e., $M=L = J + (d-1)\log_2 J$. 
The range $\mu = 0,...,J$ is treated completely 
analogous as above. In the range $\mu = L+1,...$ we chose $m_\mu = 0$ and estimate the sum over the 
operator (quasi-)norms of $\id_\mu$. This gives
\begin{equation}\label{est111}
    \sum\limits_{\mu=L+1}^{\infty} \|\id_\mu\|^{\varrho} \lesssim_{\theta} 
    \sum\limits_{\mu=L+1}^{\infty}2^{-\mu r\varrho} \simeq_{\theta} (2^J J^{d-1})^{-r\varrho}\,.
\end{equation}
Let us deal with the range $\mu = J+1,...,L$\,. By Theorem \ref{thm:GelMain}, \eqref{eqn:GelMain1} we have, in contrast
to \eqref{est1} above, the improved bound
\begin{equation}\label{impr}
   c_{m_{\mu}}(\id_\mu) \simeq_{\theta} 2^{-r\mu}\Big(\frac{2^\mu + \log(\mu^{d-1}/m_\mu)}{m_\mu}\Big)^{1/q_0-1/q_1}\,,
\end{equation}
which results in
\begin{equation}\label{est3_1}
  \begin{split}
    \sum\limits_{\mu=J+1}^{L} c_{m_\mu+1}(\id_\mu)^{\varrho} 
    &\lesssim_{\theta} \sum\limits_{\mu=J+1}^{L} 2^{-\mu r \varrho}2^{(\mu-L)\kappa(1/q_0-1/q_1)\varrho}\\
    &\simeq_{\theta} 2^{-Lr\varrho} \simeq_{\theta} (2^JJ^{d-1})^{-r\varrho} \,,
  \end{split}  
\end{equation}
when plugging in $m_\mu = 2^{\mu}2^{(L-\mu)\kappa}$\,. The relations \eqref{est0}, \eqref{est111}, and \eqref{est3_1} together with \eqref{est8} imply \eqref{res2}\,.
\end{proof}

What concerns the endpoint case $r=1/q_0-1/q_1$ we are able to complement Propositions \ref{res:appl_embedding_entropy0}, \ref{res:appl_embedding_entropy} in case $p_0=p_1=2$ as follows\,. 

\begin{proposition}\label{res:appl_embedding_entropy2}
Let $0<q_0\leq 1$, $q_0<q_1 \leq 2$ and $r = 1/q_0-1/q_1$. Let further $\varrho = \min \{1,q_1\}$. For $m\geq 2$, we have
\begin{equation}\label{res21}
      m^{-r} \lesssim_{\theta} c_m(\id:s^r_{2,q_0}b(\Omega) \to s^0_{2,q_1}b(\Omega)) \lesssim_{\theta} m^{-r}(\log\log_2 m)^{r+1/\varrho}\,.
\end{equation}
\end{proposition}
\begin{proof} We insert $m_\mu = 2^J\mu^{d-1}$ into \eqref{impr} and take \eqref{est31} into account.
\end{proof}

Let us state our main result in this subsection. 

\begin{theorem}\label{mainB} Under the same restrictions as in Propositions \ref{res:appl_embedding_entropy0}, 
\ref{res:appl_embedding_entropy} and \ref{res:appl_embedding_entropy2} the bounds for the Gelfand numbers literally transfer to the corresponding embedding 
$$
      \Id:S^r_{p_0,q_0}B(\Omega) \to S^0_{p_1,q_1}B(\Omega)\,.
$$
\end{theorem}

\begin{proof} The result is a direct consequence of Propositions \ref{res:appl_embedding_entropy0}, 
\ref{res:appl_embedding_entropy} and \ref{res:appl_embedding_entropy2} together with the machinery described in the proof of \cite[Thm.\ 4.11]{Vyb06}\,. 
All that is needed for this machinery are the properties (S1)-(S5), which are satisfied by both Gelfand and dyadic entropy numbers.
\end{proof}

Let us emphasize the following special cases. In connection with Besov space embeddings of mixed smoothness the target space $L_p$ plays a 
particular role.  

\begin{corollary}\label{L2} Let $0<q \leq 1$ and $0<r<1/q-1/2$. Then we have
$$
    c_m(\Id:S^r_{2,q}B(\Omega) \to L_2(\Omega)) \simeq_{\theta} m^{-r}\quad,\quad m\in \N\,.
$$
\end{corollary}

\begin{proof} The result is a direct consequence of Theorem \ref{mainB} together with the fact that $S^0_{2,2}B(\Omega) = L_2(\Omega)$\,.
\end{proof}

Let us also state the following endpoint result as a special case of Proposition~\ref{res:appl_embedding_entropy2} (take $q_0=q$, $q_1=2$).

\begin{corollary}\label{endpoint} Let $0<q\leq 1$ and $r=1/q-1/2$. 
Then we have for $m>2$
\begin{equation*}
m^{-r} \lesssim_{\theta} c_m(\Id:S^r_{2,q}B(\Omega) \to L_{2}(\Omega)) \lesssim_{\theta} m^{-r}(\log\log_2 m)^{r+1}\,.
\end{equation*}
\end{corollary}

Let us finally mention the following non-sharp bound. 
\begin{corollary}\label{nonsharp} Let $0<q_0<p_0\leq 1$, $p_0\leq p_1\leq 2$ and $1/p_0-1/p_1<r<1/q_0-1/p_1$. \\
Then we have for $m>2$
$$
 m^{-r} \lesssim_{\theta} c_m(\Id:S^r_{p_0,q_0}B(\Omega) \to L_{p_1}(\Omega)) 
 \lesssim_{\theta} m^{-r}(\log\log_2 m)^{1/q_0-1/p_1}\,.
$$
\end{corollary}

\begin{proof} 
This follows directly from Theorem \ref{mainB} (with $q_1=p_1$) and the trivial embedding 
$S^0_{p_1,p_1}B(\Omega) \hookrightarrow L_{p_1}(\Omega)$ in case $0<p_1\leq 2$\,.
\end{proof}

\begin{remark}\label{large_s} {\em (i)} In the case of entropy numbers Vyb\'iral proved in \cite[Thm.\ 4.9]{Vyb06} that for any $\varepsilon>0$ there is a constant $C_{\varepsilon}>0$ such that 
\begin{align}\label{vyb}
 c\,m^{-(r_0-r_1)} \leq e_m(\id: s_{p_0,q_0}^{r_0}b(\Omega) \to s_{p_1,q_1}^{r_1}b(\Omega)) \leq C_\varepsilon m^{-(r_0-r_1)} (\log m)^\varepsilon\quad, \quad m \geq 2,
\end{align}
for $0<p_0\leq p_1 \leq \infty$, $0<q_0\leq q_1 \leq \infty$ in the case of small smoothness $1/p_0-1/p_1<r_0-r_1 \leq 1/q_0-1/q_1$. The latter result is a direct consequence of a sharp bound in the regime of ``large smoothness'' 
\begin{equation}\label{ls}
   r_0-r_1>\max\{1/p_0-1/p_1,1/q_0-1/q_1\} 
\end{equation}
which is
\begin{equation}\label{vyblarge}
    e_m(\id:s_{p_0,q_0}^{r_0}b(\Omega) \to s_{p_1,q_1}^{r_1}b(\Omega)) \simeq_{\theta} m^{-(r_0-r_1)}(\log m)^{(d-1)(r_0-r_1-1/q_0+1/q_1)}\,.
\end{equation}
In fact, the entropy numbers in \eqref{vyb} can be bounded from above by $e_m(\id: s_{p_0,q^{\ast}}^{r_0}b(\Omega) \to s_{p_1,q_1}^{r_1}b(\Omega))$ if $q^{\ast}\geq q_0$. Now choose 
$q_1>q^{\ast}>q_0$ such that $1/q^{\ast}-1/q_1+\varepsilon/(d-1)=r_0-r_1>1/q^{\ast}-1/q_1$ which, together with \eqref{vyblarge} and $q_0$ replaced by $q^{\ast}$, implies \eqref{vyb}. 
The gap in \eqref{vyb} will be closed in a forthcoming paper \cite{MaUl17} also in the endpoint case $r=1/q_0-1/q_1$. The correct order is $m^{-(r_0-r_1)}$\,.

{\em (ii)} Using almost literally the arguments in \cite[Thms.\ 3.18, 3.19]{Vyb06} together with Theorem \ref{thm:GelMain}, 
\eqref{eqn:Foucartetal} we can prove a 
direct counterpart of \eqref{vyblarge} for Gelfand 
numbers in the situation $\max\{p_0,q_0\} \leq 1, \max\{p_1,q_1\} \leq 2$, $p_0 \leq p_1$ and \eqref{ls}, which is 
\begin{equation}\label{gelflarge}
    c_m(\id:s_{p_0,q_0}^{r_0}b(\Omega) \to s_{p_1,q_1}^{r_1}b(\Omega)) \simeq_{\theta} m^{-(r_0-r_1)}(\log m)^{(d-1)(r_0-r_1-1/q_0+1/q_1)}\quad,\quad m\geq 2\,.
\end{equation}
However, the argument described after \eqref{vyblarge} does not apply for Gelfand numbers. Consider the situation of determining the Gelfand numbers 
of the embedding $c_m(\id: s^r_{2,1}b(\Omega) \to s^0_{2,2}b(\Omega))$ for $r<1/2$. 
We choose $q^\ast>1$ such that $1/q^\ast - 1/2 = r$. The problem reduces to the Gelfand numbers of the embedding $c_m(\id: s^r_{2,q^\ast}b(\Omega) \to s^0_{2,2}b(\Omega))$. 
Note that this situation is not covered by \eqref{gelflarge}. Nevertheless, we may reduce further to the finite-dimensional embedding
$\id:\ell^b_{q^\ast}(\ell_2^d) \to \ell_2^b(\ell_2^d)$. The critical point here is that in case $\min\{p_0,q_0\}>1$ the entropy and Gelfand numbers of embeddings
$\id:\ell_{q_0}^b(\ell_{p_0}^d) \to \ell_{q_1}^b(\ell_{p_1}^d)$ may differ in order (already if $1<p_0=q_0<2$ and $p_1=q_1=2$, see the second case in \cite[Lem.\ 4.7]{Vyb08}. 
Applying the technique in \cite[Thm.\ 3.19]{Vyb06} (as recently done in \cite[Thm.\ 2.2]{Kien16}) produces the additional restriction $r>1/2$ which does not help in our case where $r<1/2$. 
Therefore, a counterpart of \eqref{vyb} based on a trivial embedding argument does not seem to be available for Gelfand numbers. Hence, the results in 
Propositions \ref{res:appl_embedding_entropy0}, \ref{res:appl_embedding_entropy}, \ref{res:appl_embedding_entropy2} seem to be completely new.
\end{remark}

\begin{remark}\label{SickelNguyen} Note that in Corollary \ref{nonsharp} the case $p_0=q_0$ is
excluded. In \cite{NgSi15} Nguyen and Sickel studied 
Weyl numbers of embeddings of tensor product Besov spaces ($p_0 = q_0$) with small mixed smoothness into $L_{p_1}$. Replacing $x_n$ by $c_n$ in the 
proof of \cite[Lem.\ 6.19]{NgSi15} we obtain for $0< p_0 \leq 1 < p_1 < 2$ and $1/p_0-1/p_1 < r < 1/p_0 - 1/2$ the estimate
$$
 m^{-r} \lesssim_{\theta} c_m(\Id:S^r_{p_0,p_0}B(\Omega) \to L_{p_1}(\Omega)) 
 \lesssim_{\theta} m^{-r}(\log\log_2 m)^{1/p_0-1/2}\,.
$$
\end{remark}

\begin{remark} The special case $\Id:S^r_{2,1}B(\Omega) \to L_2(\Omega)$ in Corollary \ref{L2} is of particular interest. ``Dualizing'' its proof in Proposition \ref{res:appl_embedding_entropy} 
we obtain the following sharp result for the Kolmogorov numbers, namely
$$
    d_m(\Id:S^r_{2,2}B(\Omega) \to S^0_{2,\infty}B(\Omega)) \simeq_{\theta} m^{-r}\quad,\quad m\in \N\,,
$$
provided that $0<r<1/2$.
\end{remark}
\begin{remark}
\label{rem:diffInner}
We doubt that $\log\log m$-terms in Propositions \ref{res:appl_embedding_entropy0}, 
\ref{res:appl_embedding_entropy2} and Corollaries \ref{endpoint}, \ref{nonsharp} are necessary. 
Except in Proposition \ref{res:appl_embedding_entropy}, where we actually use the sharp
mixed-norm estimate \eqref{eqn:GelMain1}, we reduce the problem to the classical finite-dimensional non-mixed situation \eqref{eqn:Foucartetal}
via standard embeddings. By doing so, we produce an additional $\log$ which becomes a $\log\log$ at the end. 
To get an improvement, one needs to refine these estimates by working directly with the mixed norms. 
An associated mixed-norm result in the spirit of \eqref{eqn:GelMain1} where the inner $\ell_2^d$-spaces 
are replaced by $\ell^d_{p_0}$ and $\ell^d_{p_1}$ would solve this issue. But this seems to be rather involved for Gelfand numbers and we leave it as an open problem. 
However, in the case of entropy numbers this is possible as we will show in a forthcoming paper, \cite{MaUl17}.
\end{remark}

\section*{Acknowledgement}
The authors wish to thank S.\ Mayer for several fruitful discussions on this paper. 
They also acknowledge the fruitful comments by G.\ Byrenheid, V.K.\ Nguyen, E.\ Novak and H.\ Rauhut on the topic 
of this paper. The authors thank the reviewers for their detailed comments, which improved the presentation of the 
paper. T.U.\ gratefully acknowledges support by the German Research Foundation (DFG) Ul-403/2-1 and the Emmy-Noether programme, Ul-403/1-1. 

\bibliographystyle{abbrv}

\def\cprime{$'$} \def\cprime{$'$}

\end{document}